\def\year{2016}\relax
\documentclass[letterpaper]{article}
\usepackage{aaai16}
\usepackage{times}
\usepackage{helvet}
\usepackage{courier}
\usepackage{verbatim}
\usepackage{amssymb}
\usepackage{amsmath}
\usepackage{array}
\usepackage{eqnarray}
\usepackage{float}
\usepackage{graphics,graphicx,epsfig,epstopdf}
\usepackage{psfrag}
\usepackage{subfigure}
\usepackage{tikz}
\usepackage{url}
\usepackage{enumitem}
\usepackage{color}
\usepackage{numprint}

\usepackage{theorem}

\usepackage{amsthm}

\usepackage{thmtools}
\usepackage{thm-restate}

\usepackage{subfigure}
\usepackage{algorithm}
\usepackage{algpseudocode}
\usepackage[normalem]{ulem}
\usepackage{booktabs}
\usepackage{appendix}

\usepackage{xspace}

\newtheorem{thm}{Theorem}

\newtheorem{lem}{Lemma}

\newtheorem{defn}{Definition}

\algtext*{EndWhile}
\algtext*{EndIf}
\algtext*{EndFor}
\algtext*{EndProcedure}

\begin{document}
%
\title{Ad auctions and cascade model: GSP inefficiency and algorithms}
\author{Gabriele Farina \and Nicola Gatti\\
Dipartimento di Elettronica, Informazione e Bioingegneria, Politecnico di Milano\\
Piazza Leonardo da Vinci, 32\\
I-20133, Milan, Italy\\
gabriele2.farina@mail.polimi.it, nicola.gatti@polimi.it}
\newif\ifshowsupplementalmaterial
\showsupplementalmaterialtrue

\newcommand{\SupplementalMaterial}{
\ifshowsupplementalmaterial
the Appendix
\else
\cite{Aggarwal2008}
\fi
}

\maketitle
\begin{abstract}
\begin{quote}
The design of the best \emph{economic mechanism} for Sponsored Search Auctions (SSAs) is a central task in computational mechanism design/game theory. Two open questions concern the adoption of user models more accurate than that one currently used and the choice between Generalized Second Price auction (GSP) and Vickrey--Clark--Groves mechanism (VCG). In this paper, we provide some contributions to answer these questions. We study  Price of Anarchy (PoA) and  Price of Stability (PoS) over social welfare and auctioneer's revenue of GSP w.r.t. the VCG when the users follow the famous \emph{cascade model}. Furthermore, we provide exact, randomized, and approximate algorithms, showing that in real--world settings (Yahoo! Webscope A3 dataset{, 10 available slots}) optimal allocations can be found in less than 1s with up to $\numprint{1,000}$ ads, and can be approximated in less than 20ms even with more than $\numprint{1,000}$ ads with an average accuracy greater than 99\%.
\end{quote}
\end{abstract}

\section{Introduction}
\label{section:introduction}
SSAs, in which a number of \emph{advertisers} bid to have their \emph{ads} displayed in some slot alongside the search results of a keyword, constitute one of the most successful applications of \emph{microeconomic mechanisms}, with a revenue of about \$50 billion dollars in the US alone in 2014~\cite{IABreport2014}. Similar models are used in many other advertisement applications, e.g. contextual advertising~\cite{VarianGoogle}.  A number of questions remain currently open in the study of effective SSAs. Among these, two are crucial and concern the adoption of the best \emph{user model} in real--world settings and the design of the best \emph{economic mechanism}. In this paper, we provide some contributions to answer these questions. 

The commonly adopted user model is the \emph{position dependent cascade} (PDC) in which the user is assumed to observe the slots from the top  to the bottom with probabilities to observe the next slot that depend only on the slots themselves~\cite{Narahari2009}. The optimal allocation can be efficiently found in a greedy fashion. The natural extension is the \emph{ad/position--dependent cascade} (APDC) model, in which the probability to observe the next slot depends also on the ads actually allocated. This model is extensively experimentally validated~\cite{Craswell,Joachims}. However, it is not known whether finding the best allocation of APDC is $\mathcal{FNP}$--hard, but it is commonly conjectured to be. \cite{aamas2013} provides an exact  algorithm that can be used in real time only with 40 ads or less, while in real--world applications the ads can be even more than $\numprint{1,000}$. The problem admits a constant--ratio $(1-\epsilon)/4$, with $\epsilon \in [0,1]$, approximation algorithm~\cite{aamas2013,cascade}, but no known algorithm is proved to lead to a \emph{truthful} economic mechanism. In addition, the known approximation algorithms can be applied in real time only with 100 ads or less~\cite{aamas2013}. More expressive models include a contextual graph, see~\cite{DBLP:conf/atal/GattiRSV15,Fotakis}, but they are Poly--$\mathcal{APX}$--hard. Hence, the APDC is considered as the candidate providing the best trade--off between accuracy and tractability.

While the GSP is still popular in many SSAs, the increasing evidence of its limits is strongly pushing towards the adoption of the more appealing VCG mechanism, which is already successfully employed in the related scenario of contextual advertising, by Google~\cite{VarianGoogle} and Facebook~\cite{HegemanFacebook}. The main drawback of the GSP is the \emph{inefficiency} of its equilibria (in terms of social welfare) w.r.t. the VCG outcome (the only known results concern PDC):  considering the whole set of Nash equilibria in full information, the PoA of the GSP is upper bounded by about $1.6$, while considering the set of Bayes--Nash equilibria the PoA is upper bounded by about $3.1$~\cite{Leme2010FOCS}. 
Furthermore, automated bidding strategies, used in practice by the advertisers to find their best bids, may not even converge to any Nash equilibrium and, under mild assumptions, the states they converge to are shown to be arbitrarily inefficient~\cite{Markakis2010WINE}. Some works study how inefficiency is affected by some form of externalities, showing that no Nash equilibrium of the GSP provides a larger revenue than the VCG outcome~\cite{Kuminov2009AAMAS,GomesWINE2009}.

\textbf{Original contributions}. Under the assumption that the users behave as prescribed by APDC, we study the PoA and PoS both over the social welfare and over the auctioneer's revenue of the GSP w.r.t. the VCG based on APDC (we study also the restricted case in which the advertisers do not overbid). We extend our analysis done for the GSP to the VCG based on PDC. Both analyses show a wide dispersion of the equilibria with PoAs in most cases unbounded, pointing out the limits of the traditional models. Furthermore, we provide a polynomial--time algorithm that dramatically shrinks auctions instances, safely discarding ads (even more than 90\% in our experiments). We also provide a randomized algorithm (that can be derandomized with polynomial cost) finding optimal allocations with probability $0.5$ in real--world instances (based on the \emph{Yahoo!} Webscope A3 dataset) with up to $\numprint{1,000}$ ads and 10 slots in less than 1s and a \emph{maximal--in--its--range} (therefore leading to truthful mechanisms if paired with VCG--like payments) approximation algorithm returning allocations with an average approximation ratio larger than about 0.99 in less than 20ms even with $\numprint{1,000}$ ads. This shows that the APDC can be effectively used in practice also in large instances.
\ifshowsupplementalmaterial
Finally, we prove our approximation algorithm to have a ratio of $1/2$ in restricted settings. However, in general settings the proof keeps to be elusive (even if we did not find any instance in which the approximation ratio is below $1/2$).
\fi

\section{Problem formulation}
	We adopt the same SSA model used in~\cite{aamas2013} with minor variations. The  model without externalities is composed of the following elements:
	\begin{itemize} 
		\item $\mathcal{A}=\{1,\dots, N\}$ is the set of ads. W.l.o.g. we assume each advertiser (also called agent) to have a single ad, so each agent~$a\in \mathcal{A}$ can be identified with ad~$a$;

		\item $\mathcal{K}=\{1,\dots, K\}$ is the set of slots $s$ ordered from the top to the bottom. We assume w.l.o.g. $K \le N$;
		
		\item $q_{a}\in [0,1]$ is the \emph{quality} of ad $a$ (i.e., the probability a user will click ad~$a$ once observed);

		\item $v_{a} \in V_a \subseteq \mathbb{R}^+$ is the value for agent~$a$ when ad~$a$ is clicked by a user, while ${\bf v}=(v_{1},\dots,v_{N})$ is the value profile;
		
		\item $\bar{v}_a$ is the product of $q_a$ and $v_a$, and can be interpreted as the expected fraction of the ad value that is actually absorbed by the users observing ad $a$;

		\item $\hat{v}_{a} \in \hat{V}_a \subseteq \mathbb{R}^+$ is the value reported by agent~$a$, while $\hat{{\bf v}}=(\hat{v}_{1},\dots,\hat{v}_{N})$ is the reported value profile;
		
		\item $\Theta$ is the set of (ordered) allocations of ads to slots, where each ad cannot be allocated in more than one slot.
	\end{itemize}

\noindent	Given an allocation $\theta$ of ads to slots, we let:
	\begin{itemize} 
		\item $\textrm{ads}(\theta) \subseteq \mathcal{A}$ to denote the subset of ads allocated in $\theta$;
		\item $\textrm{slot}_{\theta}(a)$ to denote the slot in which ad $a$ is allocated, if any;
		\item $\textrm{ad}_{\theta}(s)$ to denote the ad allocated in slot $s$.
	\end{itemize}
	
	The externalities introduced by the cascade model assume the user to have a Markovian behavior, starting to observe the slots from the first (i.e., slot $1$) to the last (i.e., slot $K$) where the transition probability from slot~$s$ to slot~${s+1}$ is given by the product of two parameters:
	\begin{itemize} 
		\item (\emph{ad--dependent externalities}) $c_{a}\in [0,1]$ is the \emph{continuation probability} of ad~$a$;
		\item (\emph{position--dependent externalities}) $\lambda_{s} \in [0,1]$ is the \emph{factorized prominence} of slot~$s$  (it is assumed $\lambda_{K} = 0$).

	\end{itemize}
	
	The click through rate $\textrm{CTR}_{\theta}(a) \in [0,1]$ of ad $a$ in allocation~$\theta$ is the probability a user will click ad $a$ and it is formally defined as
	\begin{equation*}
		\textrm{CTR}_{\theta}(a) = q_{a}\cdot \Lambda_{\textrm{slot}_{\theta}(a)} \cdot C_{\theta}(a),
	\end{equation*}
	where
	\begin{align*}
		C_{\theta}(a) &= \displaystyle\prod \limits_{s\,<\,\textrm{slot}_{\theta}(a)} c_{\textrm{ad}_{\theta}(s)},\\ \Lambda_s &= \displaystyle\prod\limits_{k\,<\,s}\lambda_{k}.
	\end{align*}
	
	\noindent Parameter $\Lambda_{s}$ is commonly called \emph{prominence}~\cite{cascade}.
	
	Given an allocation $\theta$, its social welfare is defined as:
	\begin{equation*}
		\textrm{SW}(\theta) = \sum_{\textrm{ads}(\theta)} v_{a}\cdot \textrm{CTR}_{\theta}(a) = \sum_{\textrm{ads}(\theta)} \bar{v}_{a}\cdot \Lambda_{\textrm{slot}_{\theta}(a)} \cdot C_{\theta}(a).
	\end{equation*}
	
	The problem we study is the design of an economic mechanism $\mathcal{M}$, composed of
	\begin{itemize} 
		\item an \emph{allocation function} $f:\times_{a\,\in\,\mathcal{A}} V_{i}\rightarrow\Theta$ and
		\item a \emph{payment function} for each agent $a$, $p_{a}:\times_{a\in A}\hat{V}_{a}\rightarrow\mathbb{R}$.
	\end{itemize}
	Each agent~$a$ has a linear utility	$u_{a}(v_{a},\hat{\mathbf{v}})= v_{a}\cdot \textrm{CTR}_{f(\mathbf{v})}(a)-p_{a}(\hat{\mathbf{v}})$,
	in expectation over the clicks and pays $p_{a}(\hat{\mathbf{v}})/\textrm{CTR}_{f(\hat{\mathbf{v}})}(a)$ only once its ad is clicked. We are interested in mechanisms satisfying the following properties:
	\begin{defn}
		Mechanism $\mathcal{M}$ is dominant strategy incentive compatible (DSIC)  if reporting true values is a dominant strategy for every agent (i.e., $\hat{v}_{a}=v_{a}$).
	\end{defn}
	\begin{defn}
		Mechanism $\mathcal{M}$ is individually rational (IR) if no agent acting truthfully prefers to abstain from participating to the mechanism rather than participating. 
	\end{defn}
	\begin{defn}
		Mechanism $\mathcal{M}$ is weakly budget balance (WBB) if the mechanism is never in deficit.
	\end{defn}
	\begin{defn}
		Mechanism $\mathcal{M}$ is computationally tractable when both $f$ and $p_a$ are computable in polynomial time.
	\end{defn}
	\begin{defn}
	Allocation function $f$ is maximal in its range if $f$ returns an allocation maximizing the social welfare among a given subset of allocations that is independent of the agents' reports.
	\end{defn}
	When $f$ is maximal in its range, VCG--like payments can be used, obtaining a DSIC, IR and WBB mechanism~\cite{conf/sigecom/NisanR00}.

\section{Inefficiency of GSP and VCG with PDC}
	\subsection{GSP analysis}
		In the GSP,  ads are allocated according to decreasing reported values: supposing $\hat{v}_1 \ge \cdots \ge \hat{v}_N$, the allocation function maps ad $a$ to slot $s=a$, for each $a = 1, \dots, K$. Every agent $a = 1,\dots, K-1$ is charged a price $p_a = q_{a+1}\hat{v}_{a+1}$. Agent $K$ is charged a price $p_K = 0$ if $K = N$, or $p_K = q_{K+1}\hat{v}_{K+1}$ otherwise. The following lemma holds:
		\begin{restatable}{lem}{gspnotir}
			\label{lem:gsp not ir}
			The GSP is not IR, when users follow the APDC.
		\end{restatable}
		The GSP is well known not to be DSIC. For this reason, we study the inefficiency of its Nash equilibria.
		\begin{restatable}{lem}{gspswpoaovb}
			\label{lem:gsp sw poa ovb}
			The PoA of the social welfare in the GSP when users follow APDC is unbounded.
		\end{restatable}
		\begin{restatable}{lem}{gsppoanoovb}
			\label{lem:gsp sw poa no ovb}
			The PoA of the social welfare in the GSP when users follow APDC is $\ge K$, in the restricted case the agents do not overbid.
		\end{restatable}
		\begin{restatable}{lem}{gsprevpoa}
			\label{lem:gsp rev poa}
			The PoA of the revenue in the GSP when users follow APDC is unbounded (even without overbidding).
		\end{restatable}
		\begin{restatable}{lem}{gspswpos}
			\label{lem:gsp sw pos}
			The PoS of the social welfare in the GSP when users follow APDC is 1  (even without overbidding).
		\end{restatable}
		\begin{restatable}{lem}{gsprevpos}
			\label{lem:gsp rev pos}
			The PoS of the revenue in the GSP when users follow APDC is 0  (even without overbidding).
		\end{restatable}
		
		Proofs of the previous lemmas can be found in \SupplementalMaterial{}. Lemma~\ref{lem:gsp sw poa ovb} and \ref{lem:gsp sw poa no ovb} together show a huge dispersion of the social welfare of the allocations associated with GSP equilibria: it may be arbitrarily larger than the social welfare of the APDC. Lemma~\ref{lem:gsp rev poa} and \ref{lem:gsp rev pos} show that this situation gets even worse when we turn our attention onto the auctioneer's revenue, with the revenue of GSP equilibria being arbitrarily smaller or larger than that of the APDC.
		
	\subsection{VCG with PDC analysis}
		The PDC model is very similar to the APDC: the difference lies in the fact that $c_a = 1$ for each ad $a \in \mathcal{A}$. Payments are calculated according to the rules attaining the VCG mechanism, thus providing a IR, WBB, DSIC mechanism under the assumption that users experience ad fatigue due only to the positions of the slots and not to the allocated ads. When the user's behavior is affected by ad continuation probabilities, the above mentioned properties do not necessarily apply anymore. Indeed, the following lemma holds:
		\begin{restatable}{lem}{vcgpdnotir}
			\label{lem:vcgpd not ir}
			The VCG with PDC  is neither  IR nor DSIC, when users follow APDC.
		\end{restatable}
		We study the inefficiency of Nash equilibria, VCG with PDC  not being truthful.
		\begin{restatable}{lem}{vcgpdswpoaovb}
			\label{lem:vcgpd sw poa ovb}
			The PoA of the social welfare in the VCG with PDC when users follow APDC is unbounded.
		\end{restatable}
		\begin{restatable}{lem}{vcgpdswpoanoovb}
			\label{lem:vcgpd sw poa no ovb}
			The PoA of the social welfare in the VCG with PDC when users follow APDC is $\ge K$, in the restricted case the agents do not overbid.
		\end{restatable}
		\begin{restatable}{lem}{vcgpdrevpoa}
			\label{lem:vcgpd rev poa}
			The PoA of the revenue in the VCG with PDC when users follow APDC  is unbounded (even without overbidding).
		\end{restatable}
		\begin{restatable}{lem}{vcgpdswpos}
			\label{lem:vcgpd sw pos}
			The PoS of the social welfare in the VCG with PDC when users follow APDC is 1.
		\end{restatable}
		\begin{restatable}{lem}{vcgpdrevposovb}
			\label{lem:vcgpd rev pos ovb}
			The PoS of the revenue  in the VCG with PDC when users follow APDC is 0.
		\end{restatable}
		\begin{restatable}{lem}{vcgpdrevposnoovb}
			\label{lem:vcgpd rev pos no ovb}
			The PoS of the revenue  in the VCG with PDC when users follow APDC is $\le 1$, in the restricted case the agents do not overbid.
		\end{restatable}
		
	Proofs of the previous lemmas can be found in \SupplementalMaterial{}. Again, we see a huge dispersion, both in terms of social welfare and of revenue, among the different equilibria of VCG with PDC.

\section{Algorithms}
	\subsection{\textsc{dominated--ads} algorithm}
		We present an algorithm meant to reduce the size of the problem (i.e., the number of ads), without any loss in terms of social welfare or revenue of the optimal allocation. The central observation for our algorithm is that, under certain circumstances, given two ads $a,b$  with parameters $(\bar{v}_a, c_a)$ and $(\bar{v}_b, c_b)$ respectively, it is possible to establish \emph{a priori} that, if in an optimal allocation $b$ is allocated to a slot, then ad $a$ is allocated in a slot preceding that of $b$; whenever this is the case we say that ad $a$ ``dominates'' ad $b$. As an example, consider two ads $a$ and $b$, satisfying the condition $(\bar{v}_a > \bar{v}_b) \land (c_a > c_b)$: in accordance with intuition, a simple exchange argument shows that $a$ dominates $b$. A weaker sufficient condition for deciding whether $a$ dominates $b$ is given in the following lemma and proved in \SupplementalMaterial{}.
		\begin{restatable}{lem}{dominads}
			\label{lem:chc condition}
			Let $\lambda_{\max}=\max\{\lambda_s:s\in \mathcal{K}\}$, and $B$ an upper bound%
			\footnote{
				More precisely, it is enough that $B$ is a (weak) upper bound of the quantity
				\[
					\tilde{B} = \max_i \{\lambda_i \cdot \textsc{alloc}(i + 1, K)\},
				\]
				where $\textsc{alloc}(i, K)$ is the value of the optimal allocation of the problem having as slots the set $\mathcal{K}_i = \{i, \dots, K\} \subseteq \mathcal{K}$.
			} of the value of the optimal allocation; also, let $\mathcal{D} = [0, \lambda_{\max}] \times [0, B]$. Given two ads $a,b$ with parameters $(\bar{v}_a, c_a)$ and $(\bar{v}_b, c_b)$, consider the affine function $w_{a,b}: \mathcal{D} \rightarrow \mathbb{R}$, defined as
			\begin{equation*}
				w_{a,b}(x, y) = \det\begin{pmatrix}
					x & -y & 1\\
					1  & \bar{v}_b & c_b\\
					1  & \bar{v}_a & c_a
				\end{pmatrix}
			\end{equation*}
			If the minimum of $w_{a,b}$ over $\mathcal{D}$ is greater than 0, then $a$ dominates $b$.
		\end{restatable}
		We will use the notation $a \prec b$ to denote that ad $a$ dominates ad $b$, in the sense of Lemma~\ref{lem:chc condition}. Note that $\prec$ defines a partial order over the set of ads. Since $w_{a,b}$ is an affine function defined on a convex set, it must attain a minimum on one of the corner points of $\mathcal{D}$, hence the following holds:
		\begin{lem}
			\label{lem:discard four points}
			If the four numbers $w_{a,b}(0,0)$, $w_{a,b}(0, B)$, $w_{a,b}(\lambda_{\max}, 0)$, $w_{a,b}(\lambda_{\max}, B)$ are all positive, then $a\prec b$.
		\end{lem}
		We define the $\emph{dominators}$ of an ad $a$ as the set $\textrm{dom}(a) = \{b \in \mathcal{A}: b\prec a\}$. The following lemma is central to our algorithm, as it expresses a sufficient condition for discarding an ad from the problem:
		\begin{restatable}{lem}{discard}
		 	\label{lem:chc discard condition}
			If $|\emph{\textrm{dom}}(a)| \ge K$, then ad $a$ can be discarded, as it will never be chosen for an allocation.
		\end{restatable}
		 This suggests a straightforward algorithm to determine the set of safely deletable ads, as streamlined in Algorithm~\ref{algo:chc}.
			\begin{figure}[t]
				\begin{algorithm}[H]
				  \small
				  \caption{\small \textsc{dominated--ads}}
				  \label{algo:chc}
				  \begin{algorithmic}[1]
				    \Procedure{dominated--ads}{ads, slots}
				    \State Determine $\mathcal{D}$, as defined in Lemma \ref{lem:chc condition}
				    \State For each ad $a$, compute $|\textrm{dom}(a)|$
				    \State Discard all ads $a$ having $|\textrm{dom}(a)| \ge K$
				    \EndProcedure
				  \end{algorithmic}
				\end{algorithm}
			\end{figure}
		\ifshowsupplementalmaterial
		In \SupplementalMaterial{} we show two methods, respectively named $\textsc{const--}\lambda$ and $\textsc{decouple}$, for calculating a ``good'' upper bound $B$ that is needed in Line 2 in order to determine $\mathcal{D}$. The first method (\textsc{const--}$\lambda$) has a computational complexity of $O(NK)$ time, while \textsc{decouple} runs in $O(N + K\log K)$ time and makes use of the FFT algorithm \cite{cooley1965algorithm}. Both methods require linear, i.e. $O(N + K) = O(N)$ memory. Of course, when $K$ is small enough, both algorithms can be run and the smallest upper bound can be taken as the value of $B$.
		\fi
		
		A na\"ive implementation of Line 3 tests every ad $a$ against all the other ads, keeping track of the number of dominators of $a$.
		\ifshowsupplementalmaterial
		This leads to an easily implementable $O(N^2)$ algorithm suitable for instances with $\approx \numprint{1,000}$ ads, but is not applicable in real time when $N$ grows. In \SupplementalMaterial{} we show that it is possible to compute the number of dominators of an ad $a$ in $O(\log N)$ amortized time using the dynamic fractional cascading technique \cite{mehlhorn1990dynamic}.
		Hence, the following lemma holds:

		\begin{lem}
			The \textsc{dominated--ads} algorithm runs in $O(N\log N)$ or $O(N(K + \log N))$ time, depending on the employed upper bounding strategy, and $O(N)$ memory.
		\end{lem}
	
		Finally notice that iterating the algorithm could lead to successive simplifications of the ads set, as the upper bound $B$ provided by $\textsc{const--}\lambda$ and $\textsc{decouple}$ may decrease in successive runs of \textsc{dominated--ads}. However, the maximum number of iterations is $N$, since, if no ad has been discarded, then $B$ does not change and no further ad can be discarded in future iterations.
		\fi

	\subsection{\textsc{colored--ads} algorithm}
		We present an algorithm that can determine an optimal allocation in time polynomial in $N$ and exponential in $K$. As the number of slots is generally small (typically $\le 10$), this algorithm is suitable for real--world applications. The key insight for the algorithm is that the problem of finding the allocation maximizing the social welfare can be cast to an instance of the well--known problem of finding a maximal $K$--vertex weighted simple path in an un undirected graph. Efficient algorithms for the latter problem can therefore be used to solve our problem.
		
		We introduce the definition of a \emph{right--aligned allocation}.
		\begin{defn}
			\label{def:right aligned allocation}
			Given an ordered sequence $S: a_1, \dots, a_n$ of ads, we say that the \emph{right--aligned allocation} of $S$ is the allocation mapping the ads in $S$ to the last $n$ slots, in order, leaving the first $K - n + 1$ slots vacant.
		\end{defn}
		
		Consider the complete undirected graph $\mathcal{G}$ having the available ads for vertices; every simple path $\pi: a_1, \dots, a_n$ of length $n \le K$ in $\mathcal{G}$ can be uniquely mapped to the right--aligned allocation of $\pi$, and \emph{vice versa}. Following this correspondence, we define the \emph{value} of a simple path in $\mathcal{G}$ as the value of the corresponding right--aligned allocation. In order to prove that the problem of finding a maximal $K$--vertex weighted simple path in $\mathcal{G}$ is equivalent to that of finding an optimal allocation, it is sufficient to prove that there exists (at least) one optimal allocation leaving no slot vacant, i.e. using all of the $K$ slots. To this end we introduce the following lemma, whose proof is in \SupplementalMaterial{}.
		\begin{restatable}{lem}{optalloclem}
			\label{lem:nonempty slots}
			In at least one optimal allocation all the available slots are allocated.
		\end{restatable}
		
		The problem of finding a maximal $K$--vertex weighted simple path in $\mathcal{G}$ can be solved by means of the \emph{color coding} technique~\cite{ColorCoding}. The basic idea is as follows: in a single iteration, vertices are assigned one of $K$ random colors; then, the best path visiting every color exactly once is found using a dynamic programming approach; finally, this loop is repeated a certain amount of times, depending on the type of algorithm (randomized or derandomized). In the randomized version, the probability of having missed the best $K$--path decreases as $e^{-R/{e^K}}$, where $R$ is the number of iterations. Therefore, using randomization, we can implement a $O((2e)^K N)$ compute time and $O(2^K + N)$ memory space algorithm, as streamlined in Algorithm~\ref{algo:cc}. In the derandomized version, the algorithm is exact and requires $O((2e)^K K^{O(\log K)} N \log N)$ time when paired with the derandomization technique presented in \cite{naor1995splitters}.
		
		\ifshowsupplementalmaterial\else
		\begin{figure}[t]
		\fi
				\begin{algorithm}[H]
				  \small
				  \caption{\small \textsc{colored--ads}}
				  \label{algo:cc}
				  \algblockdefx{Repeat}{EndRepeat}[1]{\textbf{repeat} #1}{\textbf{end repeat}}
				  \algtext*{EndRepeat}
				  \begin{algorithmic}[1]
				    \Procedure{colored-ads}{ads, slots}
						\Repeat{$e^K \log 2$ \textbf{times}}\Comment{50\% success probability}
							\State Assign random colors $\in \{1, \dots, K\}$ to $\mathcal{G}$'s vertices
							\hspace*{27px}{\textbf{note}: each color must be used at least once}
							\vspace{2mm}\newline
							\Comment{
								\textcolor{black!70!white}{We now construct a memoization table $\textsc{memo}[\mathcal{C}]$, reporting, for each color subset $\mathcal{C}$ of $\{1,\dots,K\}$ the maximum value of any simple path visiting colors in $\mathcal{C}$ exactly once.}\vspace{2mm}
							}
							\State $\textsc{memo}[\emptyset] \gets 0$
							\State $\mathcal{P} \gets $ powerset of $\{1,\dots,K\}$, sorted by set size
							\For{$\mathcal{C} \in \mathcal{P}-\{\emptyset\}$, in order}
								\State $\tilde{\lambda} \gets \lambda_{K-|\mathcal{C}|+1}$
								\For {\textbf{each} color $c\in\mathcal{C}$}
									\For {\textbf{each} vertex (ad) $a$ in $\mathcal{G}$ of color $c$}
										\State $\textsc{value} \gets \bar{v}_a + c_a\,\tilde{\lambda} \cdot\textsc{memo}[\mathcal{C}-\{c\}]$
										\State $\textsc{memo}[\mathcal{C}] \gets \max\{\textsc{memo}[\mathcal{C}], \textsc{value}\}$
									\EndFor
								\EndFor
							\EndFor
						\EndRepeat
						\State \textbf{return} $\textsc{memo}[\{1,\dots,K\}]$
				    \EndProcedure
				  \end{algorithmic}
				\end{algorithm}
		\ifshowsupplementalmaterial\else
			\end{figure}
		\fi
			
		Note that Algorithm~\ref{algo:cc} is just a simplified, randomized and non--parallelized sketch of the algorithm we test in the experimental section of this paper; also, for the sake of presentation, in Algorithm~\ref{algo:cc} we only return the value of the optimal allocation and not the allocation itself. We also considered the idea of using $1.3K$ different colors, as suggested by the work of \cite{huffner2008algorithm}, but we found out that for this particular problem the improvement is negligible.
		
		We conclude this subsection with some remarks about the above algorithm. First, we remark that, given the nature of the operations involved, the algorithm proves to be efficient in practice, with only a small constant hidden in the big--oh notation. Furthermore, it is worth to note that the iterations of the main loop (Lines 2 to 15) are independent; as such, the algorithm scales well horizontally in a parallel computing environment. Second, we point out an important economic property of the algorithm, that makes \textsc{colored--ads} appealing: it allows the design of truthful mechanisms when paired with VCG--like payments. While this is obviously true when the algorithm is used in its derandomized exact form, it can be proven that the truthfulness property holds true even when only a partial number of iterations of the main loop is carried out. This easily follows from the fact that the algorithm is maximal--in--its--range, searching for the best allocation in a range that does not depend on the reports of the agent. This implies that it is possible to interrupt the algorithm after a time limit has been hit, without compromising the truthfulness of the mechanism. This leads to approximate algorithms that offer convenient time--approximation trade--offs, as long as $K$ is small.
		\ifshowsupplementalmaterial
		We experimentally investigate the approximation ratio of these algorithms as a function of the number of performed iterations of the main loop in \SupplementalMaterial{}.
		\fi

	\subsection{\textsc{sorted--ads} approximate algorithm}
		While the general problem of finding an optimal allocation is difficult, polynomial time algorithms are easy to derive---as we show below---when we restrict the set of feasible allocations to those respecting a given total order $\prec_{\textrm{ads}}$ defined on the ads set. This suggests this simple approximation algorithm: first, $T$ total orders $\prec_{\textrm{ads},1}, \dots, \prec_{\textrm{ads},T}$ over the ads set are chosen; then, the optimal allocation satisfying $\prec_{\textrm{ads},i}$ is computed, for each $i$; finally, the value of the optimal allocation for the original unrestricted problem is approximated with the best allocation found over all the $T$ orders.
		The number of total orders $T$ is arbitrary, with more orders obviously producing higher approximation ratios.
		\ifshowsupplementalmaterial
		An empirical study on the time--approximation tradeoff curve of \textsc{sorted--ads} is available in \SupplementalMaterial{}.
		\fi
		
		In order to find the optimal allocation respecting the total order $\prec_{\textrm{ads}}$, we propose a simple $O(NK)$ time dynamic programming algorithm, which we name \textsc{sorted--ads}, described in Algorithm~\ref{algo:sorted}.
			\begin{figure}
				\begin{algorithm}[H]
				  \small
				  \caption{\small \textsc{sorted--ads}}
				  \label{algo:sorted}
				  \begin{algorithmic}[1]
				    \Procedure{sorted-ads}{ads, slots, $\prec_{\textrm{ads}}$}
					\State Sort the ads according to $\prec_{\textrm{ads}}$, so that ad $1$ is the minimum ad w.r.t. the given order\vspace{2mm}\newline
					\Comment{
						\textcolor{black!70!white}{We now construct a memoization table $\mathcal{T}[n, k]$, reporting, for each $1\le n\le N$ and $1\le k \le K$, the value of the best allocation that uses only ads $n, \dots, N$ and slots $k, \dots, K$ (i.e. no ads get allocated to slots $1, \dots, k-1$ and $\lambda_i = 1, \forall i < k$).}\vspace{2mm}
					}
					\State $\mathcal{T}[N, k] \gets \bar{v}_N,\quad k=1,\dots,K$\Comment{Base case}
					\For{$n = N - 1$ \textbf{downto} $1$}
						\For{$k = 1, \dots, K$}
							\If{$k < K$}
								\State $\textsc{value} \gets \bar{v}_n + \lambda_k\,c_n\,\mathcal{T}[n + 1, k + 1]$
								\State $\mathcal{T}[n, k] \gets \max\{\textsc{value}, \mathcal{T}[n + 1, k]\}$
							\Else 
								\State $\mathcal{T}[n, k] \gets \max\{\bar{v}_n, \mathcal{T}[n + 1, K] \}$
							\EndIf
						\EndFor
					\EndFor
					\State \textbf{return} $\mathcal{T}[1, 1]$
				    \EndProcedure
				  \end{algorithmic}
				\end{algorithm}
			\end{figure}
		The idea behind the algorithm is to find, for each $n = 1,\dots, N$ and $k = 1,\dots,K$, the value $\mathcal{T}[n,k]$ of the best allocation for the problem having $\mathcal{A}_n = \{1,\dots, n\} \subseteq \mathcal{A}$ as ads set and $\mathcal{K}_k = \{k, \dots, K\} \subseteq \mathcal{K}$ as slots set. The values of $\mathcal{T}[n,k]$ can be computed inductively, noticing that the associated subproblems share the same optimal substructure.
		As before, in Algorithm~\ref{algo:sorted} we only show how to find the value of the optimal allocation, and not the allocation itself.
		\ifshowsupplementalmaterial
		Also, for the sake of simplicity, the algorithm we present uses $O(NK)$ memory, but we remark that it is possible to bring the required memory down to $O(N + K)=O(N)$ while letting fast (i.e. $O(NK)$ time) reconstruction of the optimal allocation, using a technique similar to that presented in \cite{hirschberg1975linear}. See \SupplementalMaterial{} for further details.
		\fi
		
		\ifshowsupplementalmaterial
		We end this subsection with the analysis of some economical and theoretical properties of the algorithm. From the economical standpoint, we note that, if the total orders used do not depend on the reported types of the agents, \textsc{sorted--ads} is maximal--in--its--range, leading thus  to a truthful mechanism when paired with VCG--like payments. Furthermore, the resulting mechanism requires polynomial time both for the allocation and the payments. From the theoretical standpoint, we remark that it is possible to prove a $1/2$ approximation ratio for the \textsc{sorted--ads} algorithm in particular settings, e.g. when all the $\lambda$'s are constant (see \SupplementalMaterial{}). The proof in the general case, though, is left here as an open conjecture.
		\else
		We end this subsection with the analysis of some economical and theoretical properties of the algorithm. We note that, if the total orders used do not depend on the reported types of the agents, \textsc{sorted--ads} is maximal--in--its--range, leading thus  to a truthful mechanism when paired with VCG--like payments. Furthermore, the resulting mechanism requires polynomial time both for the allocation and the payments.
		\fi


\section{Experimental evaluation}

		For a better comparison of the results, we adopted the same experimental setting used in~\cite{aamas2013} and given to us by the authors. We briefly describe it, details can be found in the original paper. The experimental setting is based on \emph{Yahoo! Webscope A3} dataset. Each bid is drawn from a truncated Gaussian distribution, where the mean and standard deviation are taken from the dataset, while quality is drawn from a beta distribution. The values of $\lambda_s$ of the first 10 slots are $\{1.0,0.71,0.56,0.53,0.49,0.47,0.44,0.44,0.43,0.43\}$. We considered two scenarios, one having $K = 5$ and one having $K = 10$. In both cases we let $N\in \{50, 60, \dots, 100\} \cup \{200,300, \dots, 1000\}$. For each pair $(K,N)$, 20 instances were generated. We implemented our algorithms in the C++11 language and executed them on the OSX 10.10.3 operating system. The main memory was a 16GB 1600MHz DDR3 RAM, while the processor was an Intel Core i7--4850HQ CPU. We compiled the source with GNU \texttt{g++} version 4.9.1. Parallelization was achieved using OpenMP version 4.0.
		\ifshowsupplementalmaterial
		Please refer to \SupplementalMaterial{} for a more thorough analysis and for the complete boxplot version of the following plots.
		\fi

	\subsection{\textsc{dominated--ads} algorithm}

		We study the average number of ads that survive \textsc{dominated--ads} in Figure~\ref{fig:chc dimen}. The upper bounding strategy needed in Lemma~\ref{lem:chc condition} was implemented using a $O(NK)$ time algorithm.
		
		\begin{figure}[H]
			\centering\includegraphics[width = 0.94\linewidth]{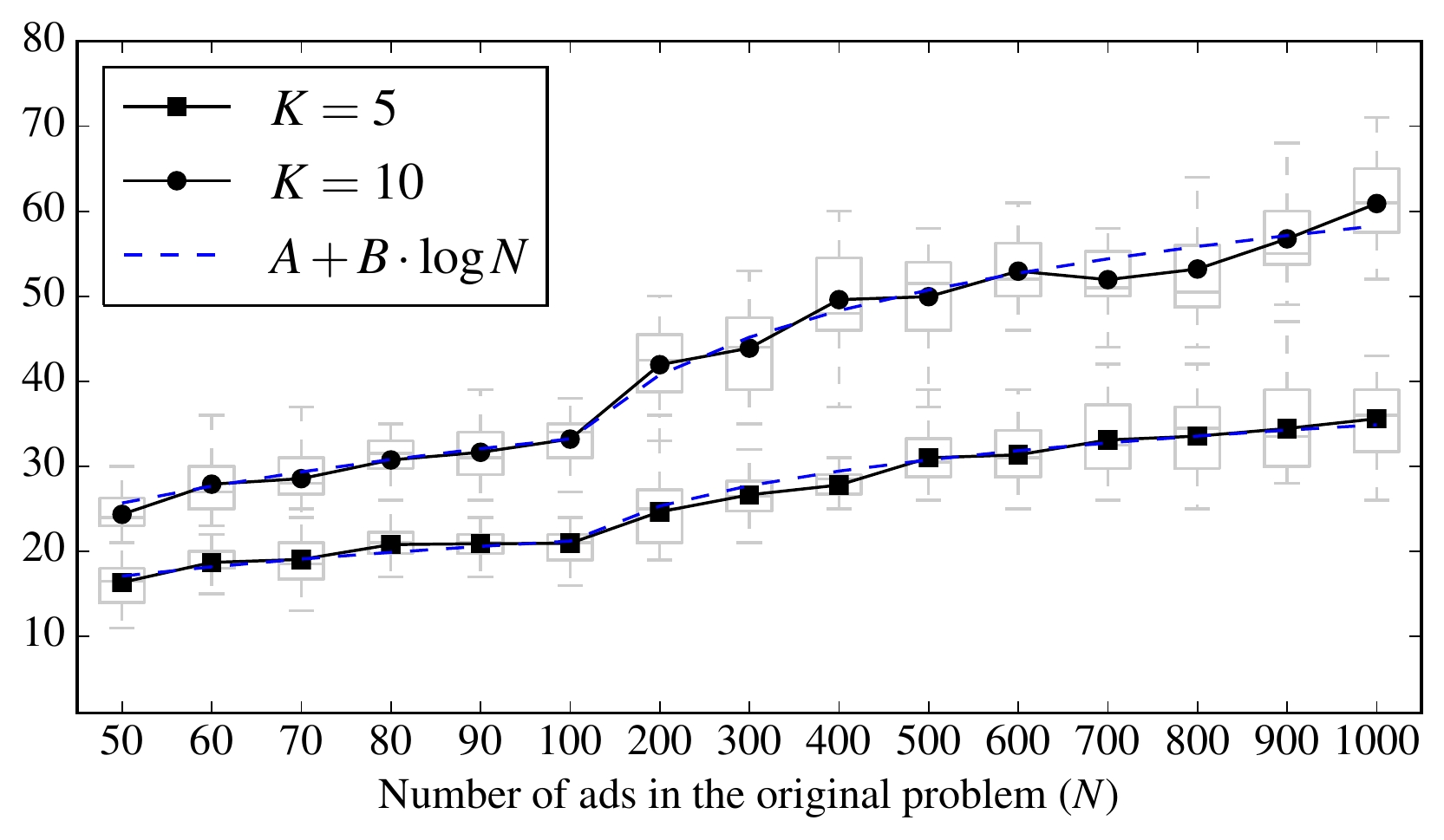}
			\caption{Average number of ads after running the \textsc{dominated--ads} algorithm.\label{fig:chc dimen}}
		\end{figure}
		
		 Notice that the number of removed ads already considerable when $N$ is small (e.g., $N \approx 50$), and that the prune ratio increases dramatically as $N$ grows (for instance, the prune ratio is approximately $96\%$ on average when $N = \numprint{1,000}$ and $K = 5$).
		Experiments show that the number of surviving ads is of the form $\tilde{N} = A + B\cdot\log N$ for suitable values\footnote{For $K = 5$, we have $A \approx -6.1, B \approx 5.9$ and $R^2 > 0.98$, where $R^2$ is the coefficient of determination of the fit.}$^,$\footnote{For $K=10$, we have $A \approx -16.9, B \approx 10.9$ and $R^2 > 0.98$.} of $A$ and $B$.
		
		In Figure~\ref{fig:chc time} we report the average running time for \textsc{dominated--ads}. The  graph shows that the running time depends quadratically on the original problem size $N$, but it remains negligible (around 20ms) even when $N = \numprint{1,000}$.
				
		\begin{figure}[H]
			\centering\includegraphics[width = 0.94\linewidth]{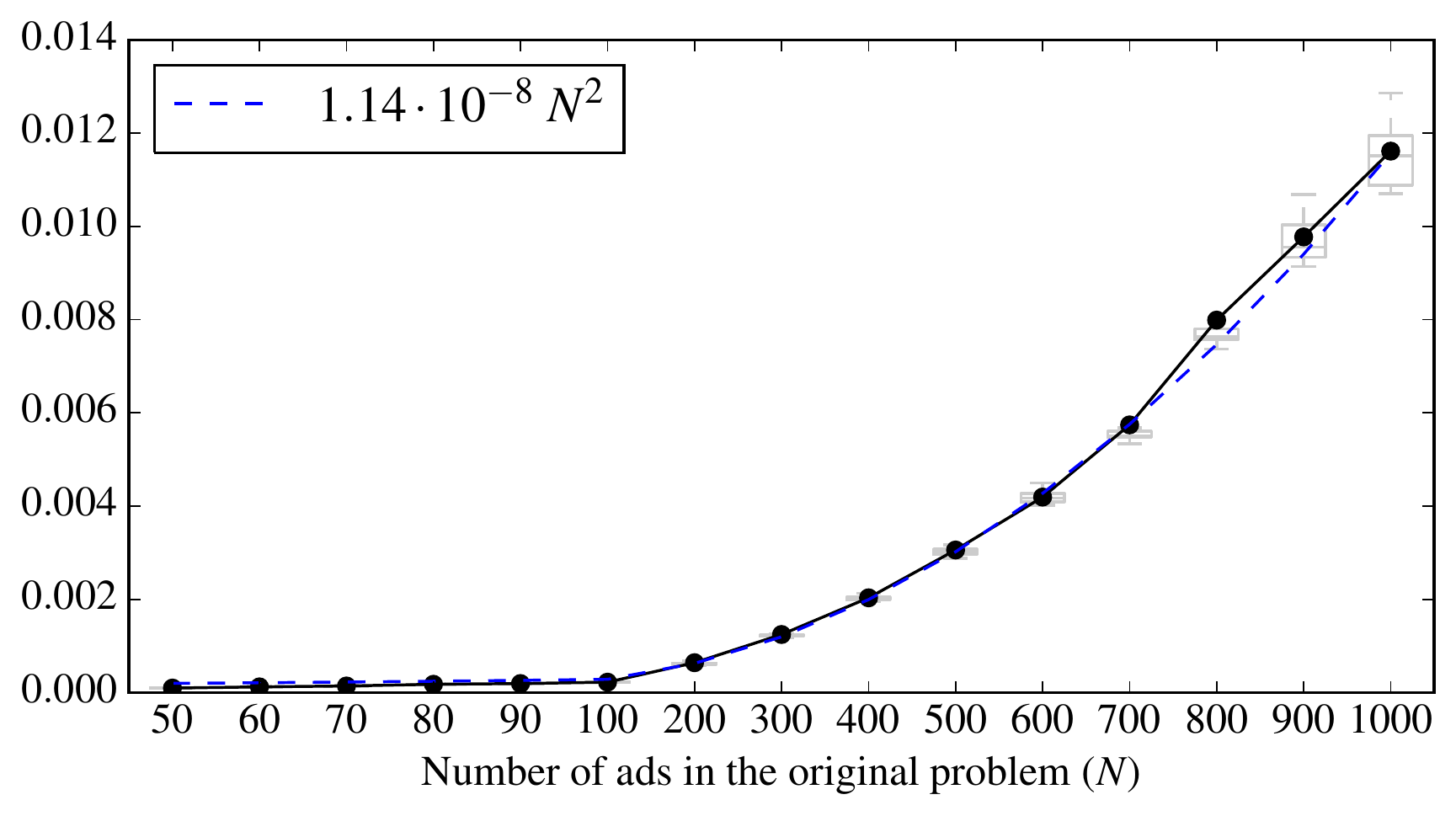}
			\caption{Avg. \textsc{dominated--ads} run. time [s]. This is independent of $K$ and grows as $\approx 1.14\cdot 10^{-8} N^2$ ($R^2 > 0.99$).\label{fig:chc time}}
		\end{figure}
			
	\subsection{\textsc{colored--ads} algorithm}
			Figure~\ref{fig:cc time} shows the average running time for \textsc{colored--ads}, implemented in its randomized version. The number of iterations was set to $e^K\log 2$, i.e. 50\% probability of finding the optimal allocation. The algorithm is run on the reduced instances produced by \textsc{dominated--ads}. As a result of the shrinking process, we point out that the running times for $N=50$ and $N=\numprint{1,000}$ are comparable. We also remark that, in order for \textsc{colored--ads} to be applicable in real--time contexts when $K = 10$, some sort of hardware scaling is necessary, as a running time of $\approx0.5$ seconds per instance is likely to be too expensive for many practical applications. When $K = 5$, though, no scaling is necessary, and the algorithm proves to be extremely efficient.

		\begin{figure}[H]
			\centering\includegraphics[width = 0.94\linewidth]{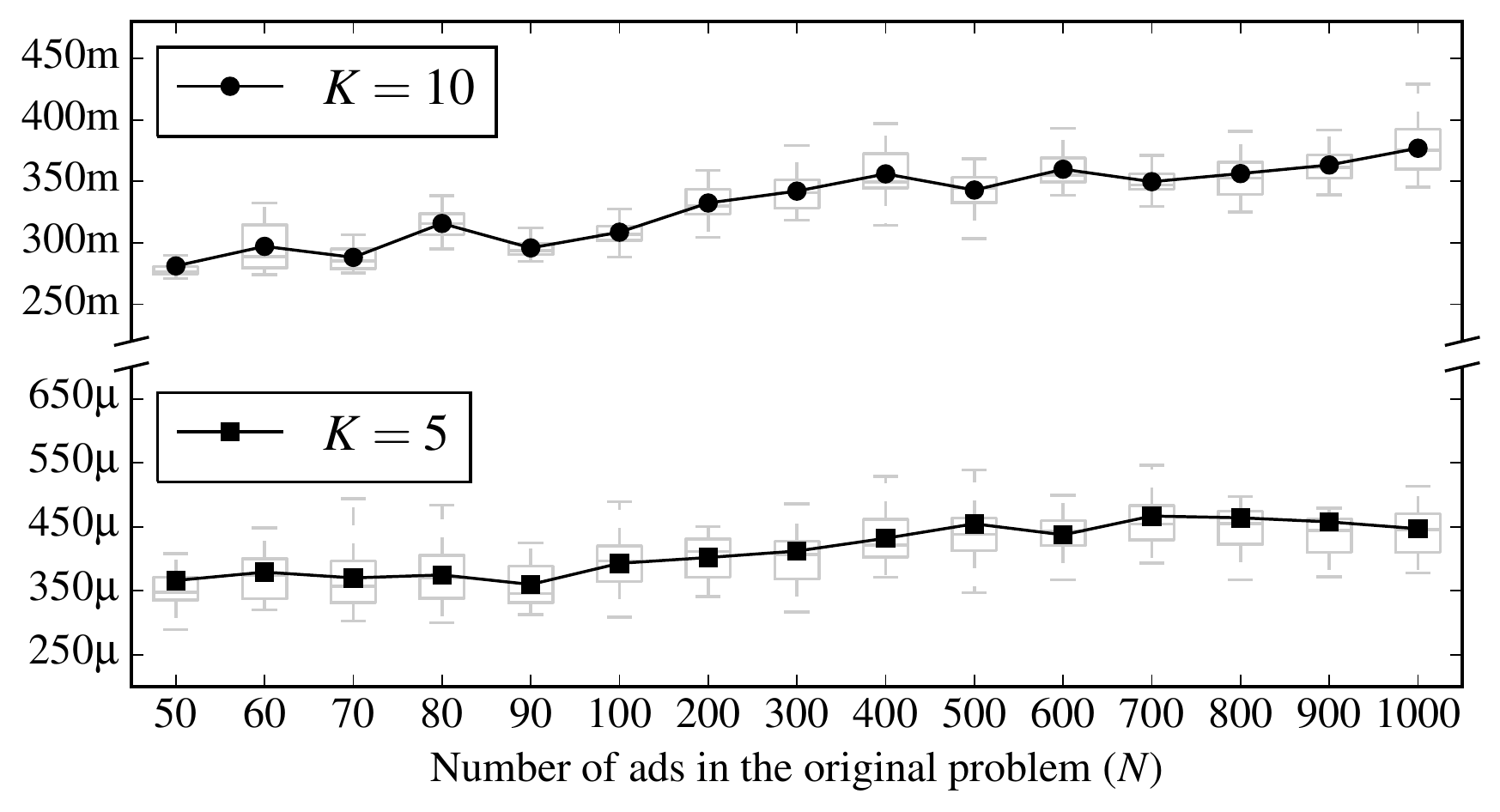}
			\caption{Avg. \textsc{colored--ads} run. time [s] ($e^K\log 2$ runs).\label{fig:cc time}}
		\end{figure}
		
	\subsection{\textsc{sorted--ads} algorithm}
		In figure~\ref{fig:sorted approx} we study the approximation ratio of \textsc{sorted-ads}. Optimal allocation values were computed using an exponential branch--and--bound algorithm, similar to that of~\cite{aamas2013}. For each shrunk problem instance, we generated $2K^3$ random orders, and used \textsc{sorted--ads} to approximate the optimal allocation. Surprisingly, even though the number of iterations is relatively low, approximation ratios prove to be very high, with values $>97\%$ in the worst case, with both medians and means always $>99\%$. Finally, in Figure~\ref{fig:sorted time} we report the corresponding computation times. These prove to be in the order of a handful of milliseconds, thus making \textsc{sorted--ads} suitable in real-time contexts even for a large number of ads.
		\begin{figure}[H]
			\centering\includegraphics[width = 0.94\linewidth]{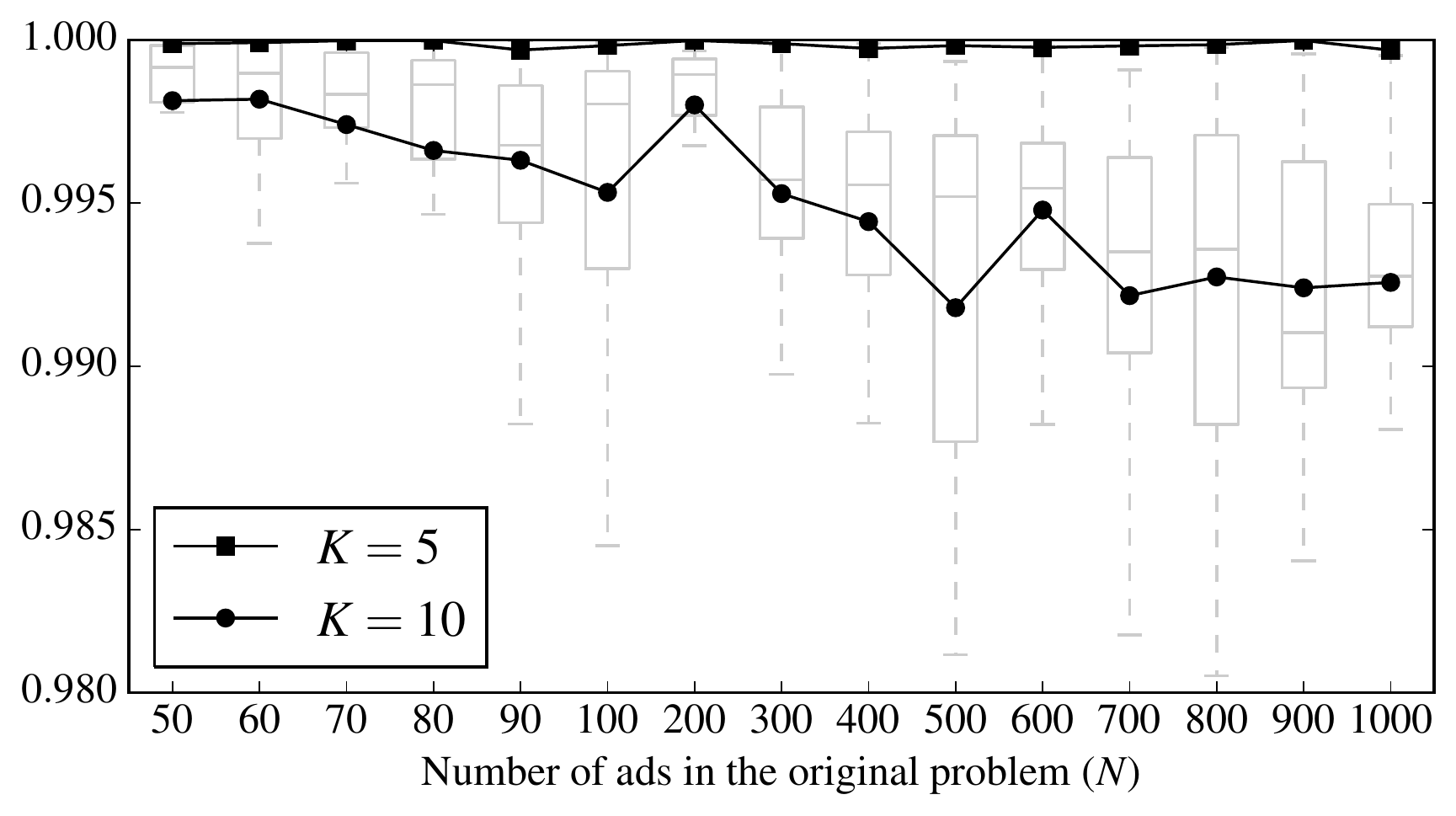}
			\caption{Avg. \textsc{sorted--ads} approx. ratio ($2K^3$ orders).\label{fig:sorted approx}}
		\end{figure}
		\begin{figure}[H]
			\centering\includegraphics[width = 0.94\linewidth]{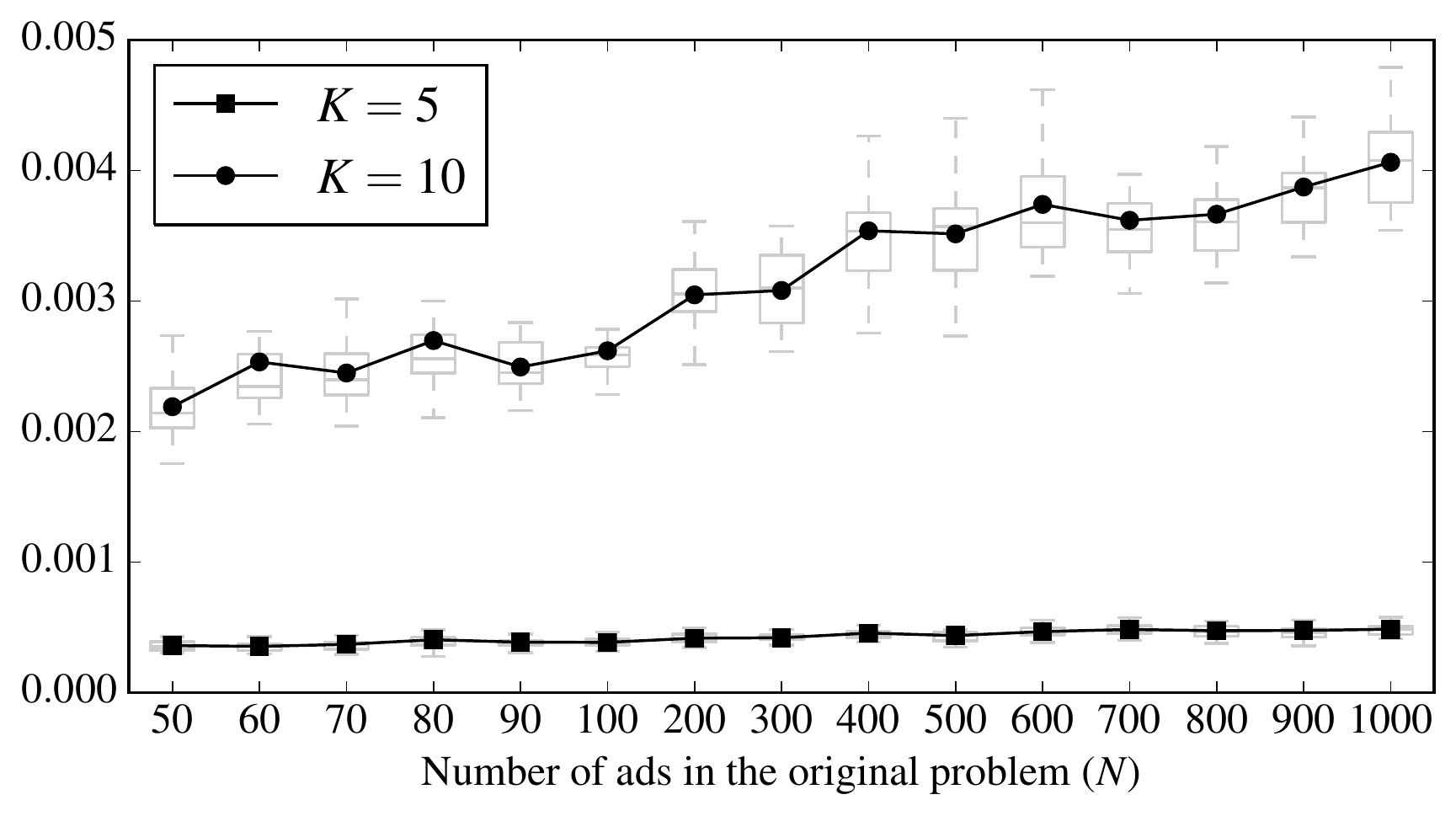}
			\caption{Avg. \textsc{sorted--ads} running time [s] ($2K^3$ orders).\label{fig:sorted time}}
		\end{figure}

\section{Conclusions and future work}
In this paper, we provide several contributions about the ad/position dependent cascade model for sponsored search auctions. This model is commonly considered the one providing the best tradeoff between accuracy and computational complexity, but no algorithm suitable for real--world applications is known. Initially, we study the inefficiency of GSP and VCG with the only position dependent cascade model w.r.t. the VCG with the ad/position dependent cascade model, analyzing PoA and PoS both of the social welfare and the auctioneer's revenue. Our results show that the inefficiency is in most cases unbounded and that the dispersion over the equilibria is extremely large, suggesting that the GSP, as well as the VCG with position dependent cascade model, presents severe limits. Subsequently, we provide three algorithms and we experimentally evaluate them with a real--world dataset. Our first algorithm reduces the size of the instances, discarding ads. Empirically, the number of non--discarded ads is logarithmic. Our second algorithm finds the optimal allocation with high probability, requiring a short time (less than 1ms when $K = 5$ even with $\numprint{1,000}$ ads), but too long for real--time applications when $K=10$. Our third algorithm approximates the optimal allocation in very short time (less than 5ms even with $\numprint{1,000}$ ads) providing very good approximations ($>$0.98). This shows that the ad/dependent cascade model can be effectively used in practice. \ifshowsupplementalmaterial
Furthermore, we prove that our approximation algorithm has a theoretical bound of $1/2$ in restricted domains and we leave open the proof that the bound holds in general settings.
\fi

In future, we will extend our analysis and our algorithms to models in which there is a contextual graph among the ad, and in which there are multiple slates.

\bibliography{citations}
\bibliographystyle{aaai}

\ifshowsupplementalmaterial
\clearpage\newpage
\begin{appendices}
	\section{Appendix}
	\vspace{.5cm}
	\section{Inefficiency of GSP and VCG with PDC}
	
	\subsection{GSP analysis}
	\gspnotir*
	\begin{proof}
		Consider a problem with $N > K \ge 2$ having:
		\begin{itemize}[nolistsep]
			\item ${v}_1 = \cdots = {v}_N = 1$,
			\item $q_1 = \cdots = q_N = 1$,
			\item $c_1 = \cdots = c_N = 0$,
			\item $\lambda_1 = \cdots = \lambda_{K-1} = 1$.
		\end{itemize}
		In the GSP model, there exists an optimal allocation mapping ad 2 to slot 2. Therefore, agent 2 pays 1, but actually gains no value since all the $c$'s are equal to 0. Hence, agent 2 has a negative payoff even though she is acting truthfully.
	\end{proof}

	\gspswpoaovb*
	\begin{proof}
		Consider a problem with $N = K = 2$ having:
		\begin{itemize}[nolistsep]
			\item ${v}_1 = \varepsilon, {v}_2 = 1$,
			\item $q_1 = q_2 = 1$,
			\item $c_1 = \varepsilon$, $c_2 = 1 - \varepsilon$,
			\item $\lambda_1 = 1$,
			\item $\hat{v}_1 = 1/\varepsilon, \hat{v}_2 = \varepsilon^2 / 2$.
		\end{itemize}
		The utilities for the bidders are
		\[
			\begin{aligned}
				&u_1 = \bar{v}_1 - q_2 \hat{v}_2 = \varepsilon - \varepsilon^2/2,\\
				&u_2 = c_1 \lambda_1 \bar{v}_2 - 0 = \varepsilon.
			\end{aligned}
		\]
		Furthermore, the GSP allocates ad 1 before ad 2, with a  social welfare of $2\varepsilon$. The truthful strategy profile is a Nash equilibrium, as if bidder 1 made a bid $< \varepsilon/2$ her new utility would be $u'_1 = \varepsilon - \varepsilon^2 < u_1$; analogously, if bidder 2 made a bid $>1$, then her utility would be $u'_2 = 0 < u_2$.
		In the VCG with APDC model, the optimal allocation has a social welfare of $1 + \varepsilon(1-\varepsilon)$. Letting $\varepsilon\rightarrow 0$, we obtain a ratio of $\frac{2\varepsilon}{1 + \varepsilon(1-\varepsilon)}\rightarrow 0$. This concludes the proof.
	\end{proof}

	\gsppoanoovb*
	\begin{proof}
		Consider a problem with $N > K \ge 2$ having:
		\begin{itemize}[nolistsep]
			\item ${v}_1 =  \cdots = {v}_N = 1$,
			\item $q_1 = \cdots = q_N = 1$,
			\item $c_1 = 0$, $c_2 = \cdots = c_N = 1$,
			\item $\lambda_1 = \cdots = \lambda_{K-1} = 1$,
			\item $\hat{v}_1 = \cdots = \hat{v}_N = 1$.
		\end{itemize}
		The allocation mapping ad $i$ to slot $i$ for each $i = 1,\dots, K$ is optimal for the GSP, resulting in a social welfare of $1$. Furthermore, the truthful strategy profile is a Nash equilibrium, since no agent would benefit from changing the reported type to a lower value. On the other hand, every optimal allocation for VCG with APDC allocates ad 1 in slot $K$, thus resulting in a social welfare of~$K$.  This concludes the proof.
	\end{proof}
	
	\gsprevpoa*
	\begin{proof}
		Consider a problem with $N > K \ge 2$ having:
		\begin{itemize}[nolistsep]
			\item ${v}_1 = 1$, ${v}_2 = \cdots = {v}_N = 1-\varepsilon$,
			\item $q_1 = \cdots = q_N = 1$,
			\item $c_1 = 0$, $c_2 = \cdots = c_N = 1$,
			\item $\lambda_1 = \cdots = \lambda_{K-1} = 1$,
			\item $\hat{v}_1 = 1$, $\hat{v}_2 = \cdots = \hat{v}_N = 0$.
		\end{itemize}
		The allocation mapping ad $i$ to slot $i$ for each $i = 1, \dots, K$ is optimal for the GSP. The revenue of the auctioneer is $0$, each bidder been charged a payment of 0. The above strategy profile is a Nash equilibrium, as the first bidder obviously has no incentive in decreasing its reported type, while each other bidder can be allocated in the first slot only bidding a value larger than or equal to $1$, but this would lead a payment $1> 1-\epsilon$, generating thus a strictly negative utility. The VCG with APDC admits an optimal allocation featuring ad 1 in the last slot (slot $K$). In this case, bidder 1 is charged a payment $1-\varepsilon$, leading the auctioneer to have a strictly positive revenue. This concludes the proof.
	\end{proof}
	
	\gspswpos*
	\begin{proof}
		Since the SW is calculated in the same way, independently on the allocation function, the PoS of the SW is $\ge 1$. It is then sufficient to show that in at least one circumstance the allocation of GSP and VCG with APDC coincide. To this end, consider a problem with $N > K \ge 2$ having:
				\begin{itemize}[nolistsep]
					\item ${v}_1 = \cdots = {v}_N = 1$,
					\item $q_1 = \cdots = q_N = 1$,
					\item $c_1 = \cdots = c_N = 1$,
					\item $\lambda_1 = \cdots = \lambda_{K-1} = 1$,
					\item $\hat{v}_1 = \cdots = \hat{v}_N = 1$.
				\end{itemize}
		It is easy to see that the allocation mapping ad $i$ to slot $i$ for each $i = 1, \dots, K$ is optimal with both the mechanisms. Furthermore, it is a Nash equilibrium, since no agent benefits from misreporting her true type.  This concludes the proof.
	\end{proof}
		
	\gsprevpos*
	\begin{proof}
		Consider a problem with $N = K = 2$ having:
		\begin{itemize}[nolistsep]
			\item ${v}_1 = 1$, ${v}_2 = 1/3$,
			\item $q_1 = q_2 = 1$,
			\item $c_1 = 1, c_2 = 1/2$,
			\item $\lambda_1 = 1$,
			\item $\hat{v}_1 = 1$, $\hat{v}_2 = 1/3$.
		\end{itemize}
		The GSP allocates ad $1$ before ad $2$, with a revenue of $1/3$. Notice that the above strategy profile constitutes a Nash equilibrium, as bidder 1 can only decrease its utility from $2/3$ to $1/2$ when ad $1$ is allocated below ad $2$, while bidder 2 can only decrease her utility from $1/3$ to $-2/3$ if ad $2$ is allocated in the first slot. While the APDC allocates ad $1$ before ad $2$ as well, the revenue for the system is $0$, since bidder 2 is not charged any payment (not being subject to competition for the second slot), and bidder 1 is charged $\bar{v}_2 - \lambda_1 c_1 \bar{v}_2 = 0$.  This concludes the proof.
	\end{proof}
	
	\subsection{VCG with PDC analysis}
	\vcgpdnotir*
	\begin{proof}
		Consider a problem with $N > K \ge 2$ having:
		\begin{itemize}[nolistsep]
			\item ${v}_1 = \cdots = {v}_N = 1$,
			\item $q_1 = \cdots = q_N = 1$,
			\item $c_1 = \cdots = c_N = 0$,
			\item $\lambda_1 = \cdots = \lambda_K = 1$.
		\end{itemize}
		In the VCG with PDC, there exists an optimal allocation mapping ad 2 to slot 2. Therefore, agent 2 pays 1, but actually she gains no value since all the $c$'s are equal to 0. Hence, agent 2 has a negative payoff even though she bids truthfully and the mechanism is not IR nor DSIC.  This concludes the proof.
	\end{proof}

	\vcgpdswpoaovb*
	\begin{proof}
		Consider a problem with $N = K = 2$ having:
		\begin{itemize}[nolistsep]
			\item ${v}_1 = 0, v_2 = 1$,
			\item $q_1 = q_2 = 1$,
			\item $c_1 = 0$, $c_2 = 1$,
			\item $\lambda_1 = 1/2$,
			\item $\hat{v}_1 = 4, \hat{v}_2 = 0$.
		\end{itemize}
		The utilities for the bidders are both 0 in the VCG with PDC model. The VCG with PDC allocates ad 1 before ad 2, as $q_1\hat{v}_1 + \lambda_1 q_2\hat{v}_2 > q_2\hat{v}_2 + \lambda_1 q_1\hat{v}_1$. The above strategy profile is a Nash equilibrium, as bidder 1 does not gain more by reducing her bid, while, if bidder 2 made any bid $>4$, her utility would be $\bar{v}_2 - (1-\lambda)q_1\hat{v}_1 < 0$. The social welfare produced by the VCG with PDC  is 0. On the other hand, the VCG with APDC  allocated ad 2 before ad 1, producing a strictly positive social welfare.  This concludes the proof.
	\end{proof}

	\vcgpdswpoanoovb*
	\begin{proof}
		The same setting used in the proof of Lemma~\ref{lem:gsp sw poa no ovb} applies here.
	\end{proof}

	\vcgpdrevpoa*
	\begin{proof}
		Consider a problem with $N = K \ge 2$ having:
		\begin{itemize}[nolistsep]
			\item ${v}_1 = \cdots = v_N = 1$,
			\item $q_1 = \cdots = q_N = 1$,
			\item $c_1 = \cdots = c_N = 1$,
			\item $\lambda_1 = \cdots = \lambda_{K-1} = 1/2$,
			\item $\hat{v}_1 =  \cdots =  \hat{v}_N = 1$.
		\end{itemize}
		Any allocation is optimal for the VCG with PDC, charging each bidder a payment of zero. Furthermore, the above strategy profile is a Nash equilibrium for the VCG with PDC, since in any allocation the utility of each bidder is the same. Thus, the revenue of the auctioneer with VCG with PDC is 0. On the other hand, in the VCG with APDC, all the bidders except the one allocated in slot $K$ are charged a strictly positive payment. This concludes the proof.
	\end{proof}

	\vcgpdswpos*
	\begin{proof}
		The same setting used in the proof of Lemma~\ref{lem:gsp sw pos} applies here.
	\end{proof}

	\vcgpdrevposovb*
	\begin{proof}
		Consider a problem with $N = K = 2$ having:
		\begin{itemize}[nolistsep]
			\item ${v}_1 = 1, v_2 = 0$,
			\item $q_1 = q_2 = 1$,
			\item $c_1 = 1$, $c_2 = 1/2$,
			\item $\lambda_1 = 1/2$,
			\item $\hat{v}_1 = 1, \hat{v}_2 = 1/2$.
		\end{itemize}
		Both VCG with PDC and VCG with APDC allocate ad 1 before ad 2. With VCG with PDC, bidder 1 has a utility of $3/4$, while bidder 2 has a utility of $0$. Notice that the above strategy profile is a Nash equilibrium of the VCG with PDC: if bidder 1 made a bid $<1/2$, her utility would be $1/4$; if bidder 2made a bid $>1$, her utility would be $-1/2$. The revenue for VCG with PDC is $1/4$, while that of VCG with APDC is $0$.  This concludes the proof.
	\end{proof}

	\vcgpdrevposnoovb*
	\begin{proof}
		Consider a problem with $N > K \ge 2$ having:
		\begin{itemize}[nolistsep]
			\item ${v}_1 = \cdots = v_N = 1$,
			\item $q_1 = \cdots = q_N = 1$,
			\item $c_1 = \cdots = c_N = 1$,
			\item $\lambda_1 = \cdots = \lambda_{K-1} = 1/2$,
			\item $\hat{v}_1 = \cdots = \hat{v}_N = 1$.
		\end{itemize}
		Notice that the strategy is a Nash equilibrium of the VCG with PDC, as no bidder can overbid, and no bidder benefits from lowering her reported type. Furthermore, since all the ad continuation probabilities are equal to 1, the revenues of the two systems are equal, thus proving the lemma.  This concludes the proof.
	\end{proof}

\section{\textsc{dominated--ads} algorithm}
	\dominads*
	\begin{proof}
		Without loss of generality, we introduce a fictitious slot, say $K+1$, with $\lambda_K=0$, representing the case in which an ad is not displayed. By simple calculations, it can be observed that, given an allocation $\theta$ and two allocated ads $a,b$ such that $b$ is allocated above $a$ (and $a$ may be allocated in slot $K+1$), swapping $a$ and $b$ leads to an improvement of the social welfare when the following condition holds:
		\[
		\overline{v}_a + c_a D + \overline{v}_b c_a E > \overline{v}_b + c_b D + \overline{v}_a c_b E
		\]
		where:
		\begin{itemize}[nolistsep, itemsep=2mm]
		\item $D = \displaystyle\frac{\sum\limits_{s\in\{\textrm{slot}_\theta(b),\ldots, \textrm{slot}_\theta(a)\}}q_{\textrm{ad}_\theta(s)}v_{\textrm{ad}_\theta(s)}C_\theta(\textrm{ad}_\theta(s))\Lambda_s}{C_\theta(b) \Lambda_{\textrm{slot}_\theta(b)}}$,
		\item $E = \displaystyle\frac{C_\theta(a)\Lambda_{\textrm{slot}_\theta(a)}}{C_\theta(b)\Lambda_{\textrm{slot}_\theta(b)}}$.
		\end{itemize}
		Notice that both $D$ and $E$ do not depend on the parameters ($v,q,c$) of $a$ and $b$.
		
		Ad $a$ will be always swapped with ad $b$ independently of $\theta$ (i.e., independently of the slots in which they are allocated and of how the other ads are allocated) and therefore ad $a$ dominates ad $b$, if the above inequality holds for every feasible value of $D$ and $E$. More precisely, $D$ is lower bounded by 0 and upper bounded by the value of the best allocation (i.e., $B$), while $E$ is lower bounded by 0 and upper bounded by $\lambda_{\max}$. Rearranging the above inequality and replacing $D$ with $y$ and $E$ with $x$, we obtain: 
		\[
		x (\overline{v}_b c_a - \overline{v}_a c_b) + y (c_a - c_b) + \overline{v}_a - \overline{v}_b = w_{a,b}(x,y).
		\]
This concludes the proof.
	\end{proof}
	
	\subsection{\textsc{const--$\lambda$} algorithm}
		The idea behind the \textsc{const--$\lambda$} algorithm is that an upper bound for $\textsc{alloc}(k, K)$ is the optimal allocation of the problem where all the $\lambda$ values are replaced with $\lambda_{\max}$. This value can be (exactly) computed in $O(NK)$ time using $\textsc{sorted-ads}$, by sorting on decreasing $\bar{v}/(1-\lambda_{\max}\,c)$ 
	\subsection{\textsc{decouple} algorithm}
		The idea behind the \textsc{decouple} algorithm is that
		\[
			\begin{array}{ll}
				\textsc{alloc}(k, K) &\ge \textsc{ub}_\text{decouple}(k)\\[2mm]
				&= \bar{v}^*_1 + \bar{v}^*_2 \, c^*_1\, \lambda_k + \cdots \\
				&\quad +\ \bar{v}^*_{K-k+1}\, (c^*_1 \cdots c^*_{K-k})\,( \lambda_k \cdots \lambda_{K-1})
			\end{array}
		\]
	where the $\bar{v}^*$ values represent the $\bar{v}$ values of the ads, ordered in decreasing order (i.e. $\bar{v}^*_1$ is the maximum among all the $\bar{v}$ values); same for the $c^*$ values. Notice that we can easily compute $\bar{v}^*_1, \dots, \bar{v}^*_K$ in $O(K \log K + N)$ time using a linear selection algorithm (e.g. introselect) and then sorting the topmost $K$ values with an efficient sort algorithm (the same holds for the $c^*$ values).
	
	In order to efficiently compute $\tilde{B}$, we need to compute the values $\textsc{ub}_\text{decouple}(1), \dots, \textsc{ub}_\text{decouple}(K)$ in time sub-quadratic in $K$, as in the worst case $K = N$. We define
	\[
		a_k = \bar{v}^*_k\,(c^*_1 \cdots c^*_{k-1}),  \quad b_k = (\lambda_1 \cdots \lambda_{k-1}).
	\]
	In particular, $a_1 = \bar{v}^*_1$ and $b_1 = 1$. Notice that we can assume without loss of generality that $b_k \neq 0,\ k = 1,\dots,K$. We also let
	\[
		c_k = b_k\,\textsc{ub}_\text{decouple}(k).
	\]
	
	It is easy to prove that the vector $\mathbf{c} = (c_K, \dots, c_1)$ is the (discrete) convolution of vectors $\mathbf{a} = (a_1, \dots, a_K)$ and $\mathbf{b} = (b_K,\dots, b_1)$. Indeed, for every $k \in \{1,\dots,K\}$:
	\[
		\begin{aligned}
			(\mathbf{a} * \mathbf{b})_k &= \sum_{i = 1}^{k} \mathbf{a}_{i}\,\mathbf{b}_{k-i+1} = \sum_{i = 1}^{k} a_i\,b_{K-k+i}\\
										&= \sum_{i = 1}^{k} \bar{v}^*_i\,(c^*_1 \cdots c^*_{i-1})\,(\lambda_1\cdots\lambda_{K-k+i-1})\\
										&= b_{K-k+1}\textsc{ub}_{\text{decouple}}(K-k+1)\\
										&= c_{K-k+1}\\
										&= \mathbf{c}_{k}.
		\end{aligned}
	\]
	
	This means that we can determine $\mathbf{c}$ in $O(K \log K)$ time using the FFT algorithm \cite{cooley1965algorithm}. Once we have $\mathbf{c}$, it's trivial to compute $\textsc{ub}_\text{decouple}(k)$ for $k = 1,\dots, K$ in $O(K)$ time. The final complexity for the algorithm is then $O(K\log K + N)$.
	
	\subsection{Fast dominance computation}
		In order to show how to compute $|\text{dom}(a)|$ for each ad $a$, we begin with this simple lemma:
		\begin{lem}
			\label{lem:w total order}
			For any fixed $(x, y) \in \mathcal{D}$, the order $\prec_{(x,y)}$ over the set of ads, defined as $$a \prec_{(x,y)} b \ \Longleftrightarrow\ w_{a,b}(x,y) > 0,$$ is a total order\footnote{We assume a deterministic tie-breaking condition is used to handle the case when $w_{a,b}(x,y) = 0$. Hence, from now on we will assume without loss of generality that $w_{a,b}(x,y) \neq 0$ for all $(x,y)\in\mathcal{D}$ and $a,b \in \mathcal{A}$}.
		\end{lem}
		\begin{proof}
			We have, by the properties of the determinant:
			\[
			\begin{aligned}
				w_{a,b}(x, y) &= \det\begin{pmatrix}x & -y & 1\\1 & \bar{v}_b & c_b\\1 & \bar{v}_a & c_a \end{pmatrix}\\
							  &= \det\begin{pmatrix}
							  			x & -y& 1\\
							  			1-c_b\,x & \bar{v}_b + c_b\,y & 0\\
							  			1-c_a\,x & \bar{v}_a + c_a\,y & 0
							  \end{pmatrix}\\
							  &= (1-c_b\,x)(\bar{v}_a + c_a\,y) - (1-c_a\,x)(\bar{v}_b+c_b\,y)
			\end{aligned}
			\]
			Hence, $w_{a,b}(x, y) > 0$ if, and only if,
			\[
				\frac{1-c_b\,x}{\bar{v}_b + c_b\,y} > \frac{1-c_a\,x}{\bar{v}_a + c_a\,y}
			\]
			
			As we assumed $x$ and $y$ to be fixed parameters, we see that $\prec_{(x,y)}$ defines a total order over the set of ads.
		\end{proof}

		Lemma~\ref{lem:w total order} and Lemma~\ref{lem:discard four points} together prove that the domination partial order, $\prec$, is the intersection of the four total orders $\prec_{(0,0)}$, $\prec_{(0,B)}$, $\prec_{(\lambda_{\max}, 0)}$ and $\prec_{(\lambda_{\max}, B)}$. This immediately suggests the idea of working with ranks; in particular, we define $\rho_{(x,y)}(a)$ as the rank of ad $a$ in the total order $\prec_{(x,y)}$, and introduce the 4--dimensional vector
		\[
			\mathbf{rk}(a) = \begin{pmatrix}\rho_{(0,0)}(a) \\ \rho_{(0,B)}(a) \\ \rho_{(\lambda_{\max}, 0)}(a) \\ \rho_{(\lambda_{\max}, B)}(a)\end{pmatrix}.
		\]
		
		The condition of Lemma~\ref{lem:discard four points} is then equivalent to the following:
		\begin{lem}
			\label{lem:cw rank}
			Given two ads $a$ and $b$, $\mathbf{rk}(a) < \mathbf{rk}(b)$ iff $a \prec b$.
		\end{lem}
		Here $<$ denotes component--wise comparison. Using Lemma~\ref{lem:cw rank}, we see that
		\[
			\text{dom}(a) = \{b \in \mathcal{A}: \mathbf{rk}(b) < \mathbf{rk}(a)\}
		\]
		In other words, we can cast the problem of counting the number of ads dominating $a$ to a geometric dominance counting problem in four dimensions. This problem is very well-studied in the literature, with many results showing algorithms computing dominance counts for every point in the point set in $O(N\ \text{poly\,log}\,N)$ time. For instance, we could use a 4--dimensional range tree \cite{Bentley1979244} to calculate $|\text{dom}(a)|$ for every ad $a$ in batch in $O(N\log^4 N)$ time.
		
		To lower the exponent of the $\log$ factor, we exploit the structure of the problem a bit further. We define the two sets:
		\[
			\begin{array}{c}
				\text{dom}^-(a) = \{b\in\mathcal{A}: b \prec a\ \land\ c_b \le c_a\}\\
				\text{dom}^+(a) = \{b\in\mathcal{A}: b \prec a\ \land\ c_b > c_a\}
			\end{array}
		\]
		It is trivial to verify that $\text{dom}^-(a) \cup \text{dom}^+(a) = \text{dom}(a)$ and $\text{dom}^-(a) \cap \text{dom}^+(a) = \emptyset$, i.e. that $\text{dom}^-(a)$ and $\text{dom}^+(a)$ are a partition of $\text{dom}(a)$. Hence,
		\[
			|\text{dom}(a)| = |\text{dom}^-(a)| + |\text{dom}^+(a)|.
		\]
		
		Notice that
		\[
			\frac{\partial}{\partial y}\,w_{a,b}(x,y) = c_a - c_b
		\]
		This means that when $c_b \le c_a$, $w_{(a,b)}(x, y) \ge w_{(a,b)}(x, 0)$, for all $(x,y)\in\mathcal{D}$. Analogously, when $c_b > c_a$, we have $w_{(a,b)}(x, y) \ge w_{(a,b)}(x,B)$ for all $(x,y)\in\mathcal{D}$. Therefore, it makes sense to introduce the following restricted ranks:
		\[
			\mathbf{rk}^-(a) = \begin{pmatrix}
				\rho_{(0,0)}(a)\\
				\rho_{(\lambda_{\max}, 0)}(a)		
			\end{pmatrix},\ 
			\mathbf{rk}^+(a) = \begin{pmatrix}
				\rho_{(0,B)}(a)\\
				\rho_{(\lambda_{\max}, B)}(a)		
			\end{pmatrix}
		\]
		We formalize the observation above in the next Lemma:
		\begin{lem}
			\label{lem:restricted ranks}
			Given two ads $a$ and $b$:
			\begin{itemize}
				\item $b \in \text{\emph{dom}}^-(a) \Longleftrightarrow \mathbf{rk}^-(b) < \mathbf{rk}^-(a)$;
				\item $b \in \text{\emph{dom}}^+(a) \Longleftrightarrow \mathbf{rk}^+(b) < \mathbf{rk}^+(a)$.
			\end{itemize}
		\end{lem}
		
		Assume we have a data structure $\textsc{rt}$ supporting each of the following operations in $O(\log N)$ time:
		\begin{itemize}
			\item $\textsc{rt.insert}(\mathbf{v})$ insert the 2-dimensional point $\mathbf{v}$ in the structure; point $\mathbf{v}$ is not already present.
			\item $\textsc{rt.erase}(\mathbf{v})$ deletes the 2-dimensional point $\mathbf{v}$ from the structure; point $\mathbf{v}$ is guaranteed to be present. Furthermore, after the first call to $\textsc{sc.erase}$ is made, the structure no longer accepts calls to $\textsc{ds.insert}$.
			\item $\textsc{rt.dominance}(\mathbf{v})$ counts the number of points $\mathbf{w}$ in the structure, satisfying $\mathbf{w} < \mathbf{v}$.
		\end{itemize}
		Such a data structure exists and is well-studied in the literature. It basically consists of a 2-dimensional range tree combined with the technique of \emph{dynamic fractional cascading} \cite{mehlhorn1990dynamic}.
		
		The $\textsc{rt}$ data structure can be used to compute $|\text{dom}(a)|$ for every ad $a$ in batch in $O(N\log N)$ time, as streamlined in Algorithm~\ref{algo:fast chc}.
			\begin{figure}
				\begin{algorithm}[H]
				  \small
				  \caption{\small \textsc{fast--dominated--ads}}
				  \label{algo:fast chc}
				  \begin{algorithmic}[1]
				    \Procedure{fast--dominated--ads}{ads, slots}
				    \State Determine $\mathcal{D}$, as defined in Lemma \ref{lem:chc condition}
				    \State $\textsc{rt}^+ \gets$ \textbf{new} \textsc{rt}
				    \State $\textsc{rt}^- \gets$ \textbf{new} \textsc{rt}
				    \For{\textbf{each} ad $a\in\mathcal{A}$}
				    	\State $\textsc{rt}^-\textsc{.insert}(\mathbf{rk^-}(a))$
				    \EndFor
				    \For{\textbf{each} ad $a\in\mathcal{A}$, sorted according to decreasing $c$,}
				    	\State $|\text{dom}^+(a)| \gets \textsc{rt}^+\textsc{.dominance}(\mathbf{rk^+}(a))$
				    	\State $|\text{dom}^-(a)| \gets \textsc{rt}^-\textsc{.dominance}(\mathbf{rk^-}(a))$
				   		\State \hspace{-.2mm}$|\text{dom}(a)| \gets |\text{dom}^+(a)| + |\text{dom}^-(a)|$
				    	\State $\textsc{rt}^+\textsc{.insert}(\mathbf{rk^+}(a))$
				    	\State $\textsc{rt}^-\textsc{.erase}(\mathbf{rk^-}(a))$
				    \EndFor
				    \State Discard all ads $a$ having $|\text{dom}(a)| > K$
				    \EndProcedure
				  \end{algorithmic}
				\end{algorithm}
			\end{figure}
			
\section{\textsc{colored--ads} algorithm}
	\optalloclem*
	\begin{proof}
		Let $\theta$ be any optimal allocation for the problem at hand. If it respects the condition of the lemma, the lemma is trivially true. Otherwise, we show that it is possible to build an optimal allocation $\theta^*$ from $\theta$, respecting the condition.
		
		Let $k^*$ the maximum index of any empty slot of $\theta$. If it is the last slot, take any non allocated ad, and allocate it to $k^*$. Such an ad must exist because we always assume $N \ge K$. As the value of the new allocation is not less than that of $\theta$, we conclude that the new allocation is optimal.
		
		If, on the other hand, $k^* \neq K$, we act as follows: we move all the ads allocated in slots $k^* + 1$ to $K$ one slot ahead. This operation does not decrease the value of the allocation, and at the same time leaves slot $K$ empty. We then operate as in the previous case, filling slot $K$ with any non allocated ad.
		
		We repeat this algorithm until no slot is left vacant. 
	\end{proof}
	
	\subsection{Time--approximation tradeoffs}
		{\centering\includegraphics[width=1\linewidth]{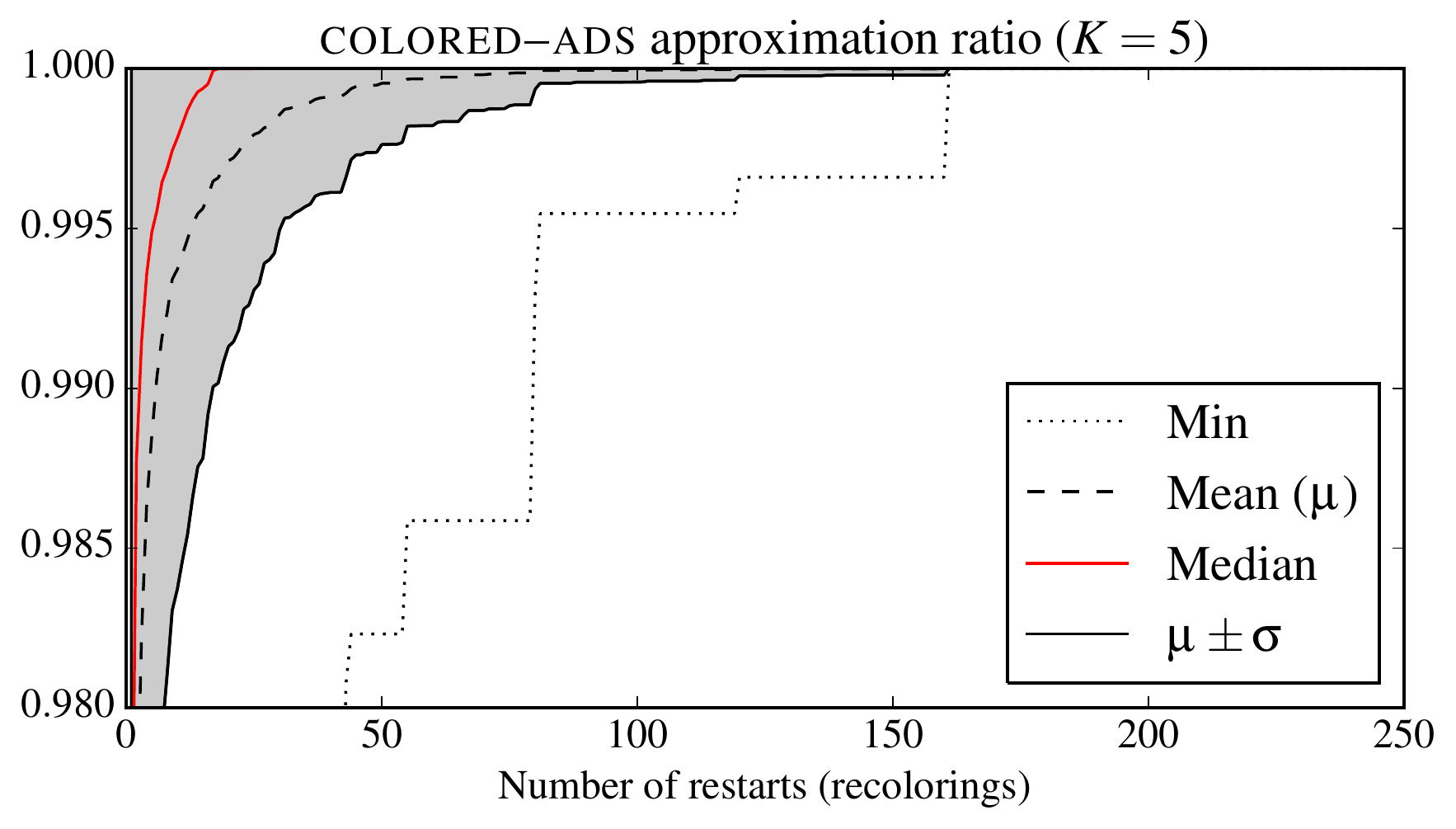}\\
		\centering\includegraphics[width=1\linewidth]{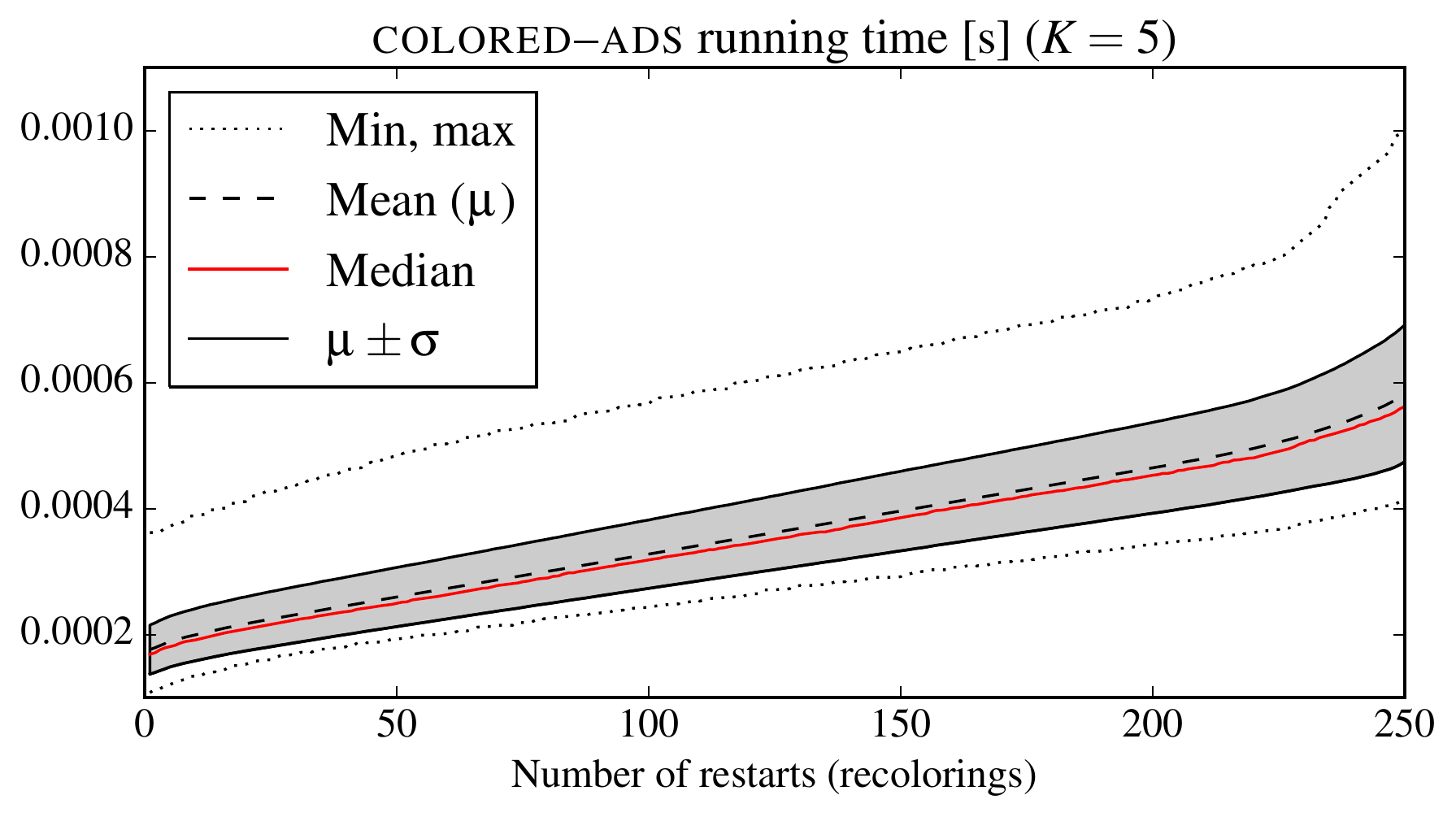}\\
		\centering\includegraphics[width=1\linewidth]{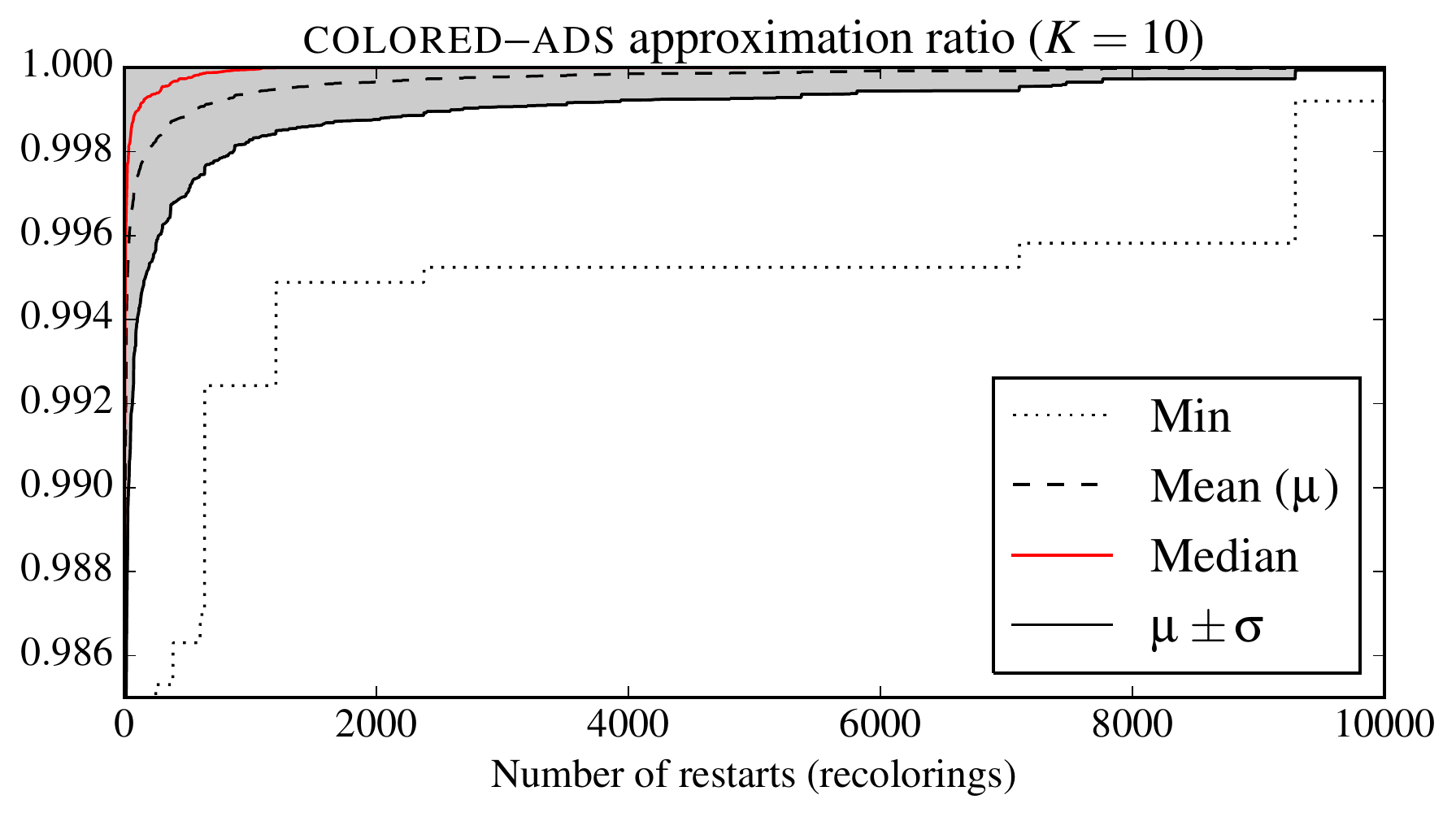}\\
		\centering\includegraphics[width=1\linewidth]{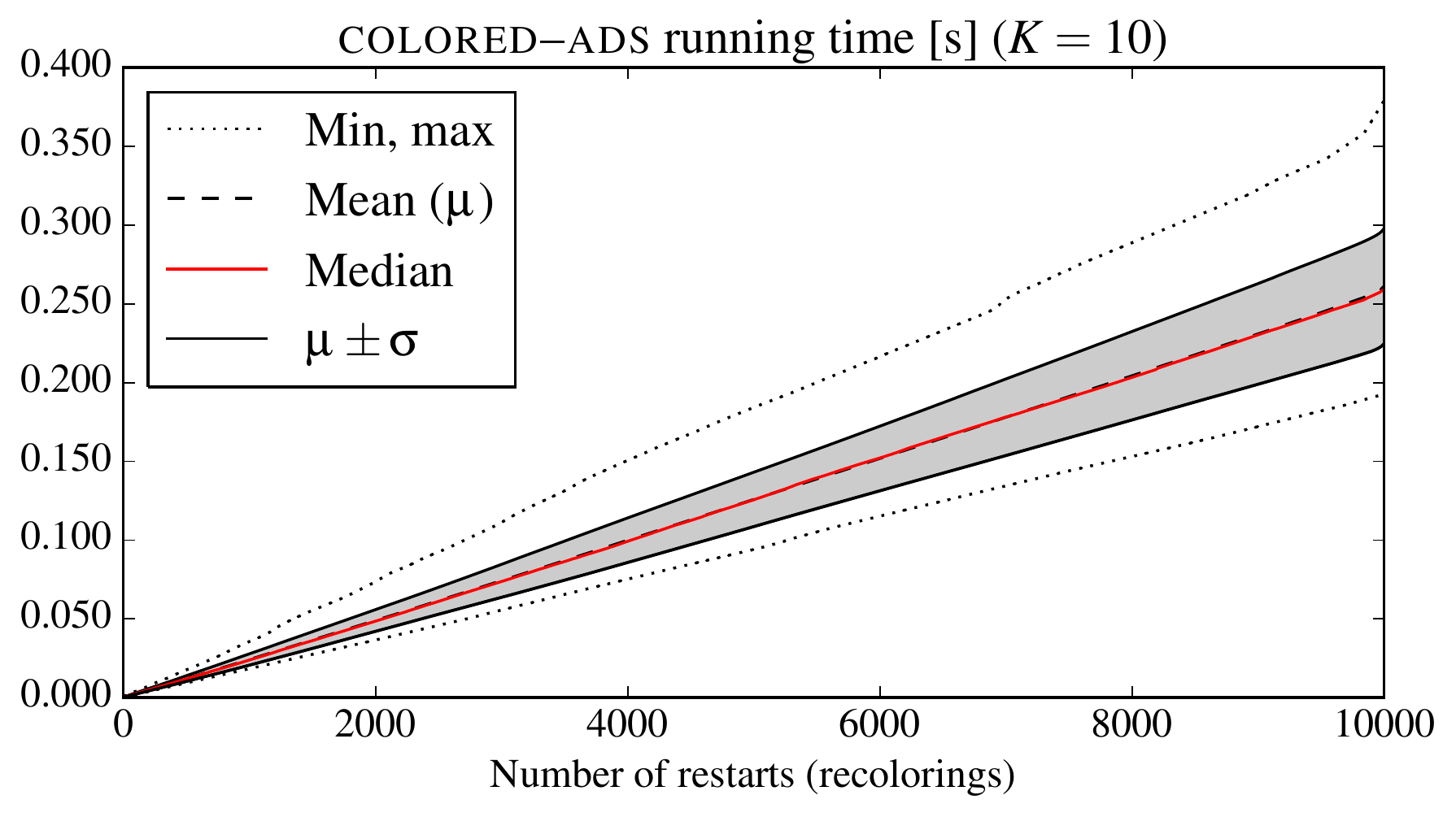}
		}
							
\section{\textsc{sorted--ads} algorithm}
	\subsection{Space--efficient algorithm}
		We modify the idea in \cite{hirschberg1975linear} to find a $O(NK)$ time, $O(N)$ memory implementation of \textsc{sorted-ads}. In particular, we need to show that there exists a $O(NK)$ time and $O(N)$ memory algorithm able to find the least slot $k^*$, respecting Definition~\ref{def:split point}.
	
	\begin{defn}
		\label{def:split point}
	 Slot $k$ is a \emph{split--point} if there exists an optimal allocation such that slots $1, \dots, k-1$ are occupied only by ads $\le N/2$ and slots $k, \dots, K$ are occupied only by ads $> N/2$.
	\end{defn}
	
	Conceptually, we scan all the possible values of $k$, one by one from $1$ up until $K$ until we find the first split--point.
	Once a tentative value for $k^*$ has been chosen, it is enough to test whether the best allocation satisfying the condition in Definition~\ref{def:split point} is an optimal allocation for the original problem. The split--point $k$ breaks the optimal allocation fulfilling the condition into two natural pieces: the first using only ads in $\{1,\dots, N/2\}$ and slots $1,\dots,k-1$, and the second piece using ads in $\{N/2 + 1, \dots, N\}$ and slots $k, \dots, K$. We let the social welfare of the first piece be $W_1$, and that of the second piece be $W_2$; also, we let $C_1$ be the product of the continuation probabilities of the ads in the first piece. The value of the optimal allocation satisfying the split--point condition is therefore
	\[
		\textrm{SW} = W_1 + \Lambda_{k-1}\,C_1\,W_2
	\]
	
	Since by definition the allocation is optimal, the second piece of the allocation must be itself an optimal allocation for the problem having $\mathcal{A} = \{N/2 + 1,\dots, N\}$, $\mathcal{K} = \{k, \dots, K\}$. Notice that we already know how to find $W_2$, for all the possible values of $k$, in $O(NK)$ time and $O(N)$ memory. We can therefore treat $W_2$ as a known constant and focus on the problem of determining the allocation of the first piece maximizing $\text{SW}$. The crucial point is that we can find the maximum of $W_1$ for each $k$ in dynamic programming. Indeed, the objective function is the same we would have in the allocation problem in which we have to allocate ads $\{1, \dots, N/2\} \cup \{X\}$ in slots $1, \dots, k$under the constraints:
		\begin{itemize}[nolistsep]
			\item $\bar{v}_X = W_2$ (notice that $W_2$ is known, while the values $c_X$ and $\lambda_k$ can be disregarded not affecting the objective function).
			\item Ad $X$ must be allocated in slot $k$.
		\end{itemize}

	\subsection{$\frac{1}{2}$--approximation ratio}

\begin{defn}
Given a total order $\prec_{\textrm{ads}}$, we define $\Theta_{\prec_{\textrm{ads}}}$ as the set of allocations $\theta$ where:
\begin{itemize}
\item for every $a,a'\in \theta$, $\emph{slot}_\theta(a)<\emph{slot}_\theta(a')$ only if $a\prec_{\textrm{ads}} a'$ (notice that the constraint is only over the allocated ads, each non--allocated ad may appear everywhere in $\prec_{\textrm{ads}}$);
\item there are no empty slots before an allocated slot.
\end{itemize}
\end{defn}

\noindent Given a $\prec_{\textrm{ads}}$, the set $\Theta_{\prec_{\textrm{ads}}}$ constitutes the range of our allocation algorithm.
\begin{defn}
We define $\theta^*_{\prec_{\textrm{ads}}}$ as the optimal allocation (maximizing the social welfare) among all the allocations in $\Theta_{\prec_{\textrm{ads}}}$.
\end{defn}

\begin{defn}
We define $\theta^*$ as the optimal allocation (maximizing the social welfare) among all the allocations in $\Theta$ and $\boldsymbol\prec^*_{\textrm{ads}}$ as $\boldsymbol\prec^*_{\textrm{ads}} = \{\prec_{\textrm{ads}}: \theta^*_{\prec_{\textrm{ads}}} = \theta^*\}$.
\end{defn}

\begin{defn}
We define $\prec_{\textrm{ads}}^N$ as: for every $a,a' \in A$, $a\prec_{\textrm{ads}}^N a'$ if and only if $\frac{\overline{v}_{a}}{1-c_a}\geq\frac{\overline{v}_{a'}}{1-c_{a'}}$ (ties are broken lexicographically). 
\end{defn}

When $\lambda_k=1$ for every $k$, we have $\prec_{\textrm{ads}}^N \in \boldsymbol\prec_{\textrm{ads}}^*$ for every value of the parameters as shown in~\cite{cascade}. 

\begin{defn}
We define $\prec_{\textrm{ads}}^R$ as the reverse of $\prec_{\textrm{ads}}^N$.
\end{defn}
We introduce a lemma that we use in our main result.

\begin{lem} 
\label{lemma:sumoveri} 
When $\lambda_k=1$ for every $k$, once $A$ is restricted to the ads appearing in $\theta^*$ and the ads are labeled such that $a$ is the $a$--th ad in the allocation, then for every $a$ we have:
\[
\frac{\overline{v}_{a}}{1-c_{a}} \geq \sum\limits_{b=a}^{K}\overline{v}_{b}\prod\limits_{h=i}^{b-1}c_{h}.
\]
\end{lem}
\begin{proof}
The proof is by mathematical induction.

\textbf{Basis of the induction}. The basis of the induction is for slots $k$ and $k-1$. We need to prove that:

\begin{equation}\label{eq:inductionbase}
\dfrac{\overline{v}_{k-1}}{1-c_{k-1}} \geq \overline{v}_{k-1} + c_{k-1}\overline{v}_{k}.
\end{equation}

\noindent Equation~(\ref{eq:inductionbase}) can be rewritten as:

\[
\dfrac{\overline{v}_{k-1}}{1-c_{k-1}} \geq \overline{v}_{k},
\]

\noindent that holds since
\[
\frac{\overline{v}_{k-1}}{1-c_{k-1}} \geq \frac{\overline{v}_{k}}{1-c_{k}}
\]
(by definition of $\prec_{\textrm{ads}}^N$) and
\[
\frac{\overline{v}_{{k}}}{1-c_{k}} \geq \overline{v}_{{k}}
\]
(by $c_{k}\in [0,1]$).

\textbf{Inductive step}. We assume that following inequality holds

\[
\dfrac{\overline{v}_{i+1}}{1-c_{i+1}} \geq \sum\limits_{j=i+1}^{k}\overline{v}_{j}\prod\limits_{h=i+1}^{j-1}c_{h},
\]

\noindent and we want to prove

\begin{equation}\label{eq:induction2}
\dfrac{\overline{v}_{a_i}}{1-c_{a_i}} \geq \sum\limits_{j=i}^{k}\overline{v}_{a_j}\prod\limits_{h=i}^{j-1}c_{a_h}.
\end{equation}

Equation~(\ref{eq:induction2}) can be written as:

\[
\dfrac{\overline{v}_{{i}}}{1-c_{{i}}} \geq \sum\limits_{j=i+1}^{k}\overline{v}_{j}\prod\limits_{h=i+1}^{j-1}c_{h},
\]

\noindent that holds since
\[
\frac{\overline{v}_{{i}}}{1-c_{{i}}} \geq \frac{\overline{v}_{i+1}}{1-c_{i+1}}
\] (by definition of $\prec_{\textrm{ads}}^N$) and 
\[\frac{\overline{v}_{i+1}}{1-c_{i+1}}\geq \sum\limits_{j=i+1}^{k}\overline{v}_{j}\prod\limits_{h=i+1}^{j-1}c_{h}\] (by assumption of the inductive step). This completes the proof.
\end{proof}

We use the above lemma to derive the bound.

\begin{thm}\label{thm:lambda1bound0.5}
When $\lambda_{k} = 1$ for every $k$, we have \[
r=\frac{\textrm{SW}(\theta^*_{\prec_{\textrm{ads}}})}{\textrm{SW}(\theta^*)}\geq \frac{1}{2}
\] for every $\prec_{\textrm{ads}}$.
\end{thm}
\begin{proof}
For the sake of presentation, we split the proof in several parts.

\textbf{Ads ranking}. Among all the possible $\prec_{\textrm{ads}}$, we can safely focus only on $\prec_{\textrm{ads}}^R$ given that this order is the worst for our problem. This property follows from the results discussed in~\cite{cascade}. More precisely, the authors show that, given any allocation $\theta$, for every $a,a'$ if $\textrm{slot}_\theta(a)<\textrm{slot}_\theta(a')$ and $\frac{\overline{v}_{a}}{1-c_{a}}>\frac{\overline{v}_{a'}}{1-c_{a'}}$ then switching $a$ with $a'$ never reduces the value of the allocation. Therefore, the social welfare with $\prec_{\textrm{ads}}^R$ is the minimum w.r.t. that one with all the $\prec_{\textrm{ads}}$. (Interestingly, when $\prec_{\textrm{ads}} \not \in \boldsymbol\prec_{\textrm{ads}}^*$, $\theta^*_{\prec_{\textrm{ads}}}$ may prescribe that the number of allocated ads is strictly smaller than the number of ads allocated in $\theta^*$.)

\textbf{Allocations set}. Given $\prec_{\textrm{ads}}^R$, we use a lower bound over the value $\theta^*_{\prec_{\textrm{ads}}^R}$ focusing only on a specific subset of allocations defined below. This subset may not contain the optimal allocation $\theta^*_{\prec_{\textrm{ads}}^R}$ and therefore  we obtain a lower bound for $r$ that may be not tight. However, in the following, we complete the result showing that the bound is tight. The allocations we focus on are of the form: $\langle z, {z-1}, \ldots , 1\rangle$ with $z\in\{1,\ldots, K\}$ where the ads are indexed on the basis on $\prec_{\textrm{ads}}^N$ (i.e., $1 \prec_{\textrm{ads}}^N 2 \prec_{\textrm{ads}}^N \ldots$). We study the value of $r$ for each possible value of $z$ under the constraints that, given $z$, it holds:
\begin{align}
\hspace{-0.15cm}\sum\limits_{i=1}^{z}\left(\overline{v}_{i}\prod\limits_{j=i+1}^{z}c_{j}\right) 	&	\geq \sum\limits_{i=1}^{z+1}\left(\overline{v}_{i}\prod\limits_{j=i+1}^{z+1}c_{j}\right)								\label{eq:condizionikplus1} \\
\vspace{-0.15cm}\sum\limits_{i=1}^{z}\left(\overline{v}_{i}\prod\limits_{j=i+1}^{z}c_{j}\right) 	&	\geq \sum\limits_{i=1}^{h}\left(\overline{v}_{i}\prod\limits_{j=i+1}^{h}c_{j}\right) 			& 	\forall h:1\leq h < z		\label{eq:condizionikminus}
\end{align}
\noindent with the exception that when $z=k$, only Constraints~(\ref{eq:condizionikminus}) are required. We notice that, for every auction instance, there always exists a value of $z$ such that the above constraints hold. This value can be found by starting from $z=1$ and increasing it until the constraints hold. 

\textbf{Ratio}. Once $z$ is fixed, a lower bound over $r$ can be obtained by using the following upper bound over the value of $\theta^*$: we compute $\theta^*$ as if $k=n$. Thus, the ratio $r$ is:
\[
r \geq \dfrac{\sum\limits_{i=1}^{z}\left(\overline{v}_{i}\prod\limits_{j=i+1}^{z}c_{j}\right)}{\sum\limits_{i=1}^{n}\left(\overline{v}_{i}\prod\limits_{j=1}^{i-1}c_{j}\right)}.
\]

\textbf{Bound with $z=1$}. In this case, we have:
\begin{align}
z							&	=	1																			\nonumber				\\
\overline{v}_{1}					&	\geq 	\overline{v}_{2}+c_2 \overline{v}_{1}													\label{eq:casok1}	\\			
r 							&	\geq \dfrac{\overline{v}_{1}}{\sum\limits_{i=1}^{n}\left(\overline{v}_{i}\prod\limits_{j=1}^{i-1}c_{j}\right)}		\label{eq:ratiokuguale1}
\end{align}
\noindent by applying Lemma~\ref{lemma:sumoveri}, we can write $r$ as:
\[
r\geq \dfrac{\overline{v}_{1}}{\sum\limits_{i=1}^{n}\left(\overline{v}_{i}\prod\limits_{j=1}^{i-1}c_{j}\right)} \geq \dfrac{\overline{v}_{1}}{\overline{v}_{1}+c_{1} \dfrac{\overline{v}_{2}}{1-c_{2}}},
\]
\noindent by using Equation~(\ref{eq:casok1}) we have:
\[
r \geq  \dfrac{\overline{v}_{1}}{\overline{v}_{1}+c_{1} \dfrac{\overline{v}_{2}}{1-c_{2}}}\geq \dfrac{\overline{v}_{1}}{\overline{v}_{1}+c_{1}\overline{v}_{1}}.
\]
\noindent The minimization of  $\frac{\overline{v}_{1}}{\overline{v}_{1}+c_{1}\overline{v}_{1}}$ can be achieved by maximizing $c_{1}$ (i.e., when $c_{1}=1$), obtaining $r\geq \frac{1}{2}$.

\textbf{Bound with $k>z>1$}. In this case, we can write $r$ as:
\begin{align*}
r&	\geq \dfrac{\sum\limits_{i=1}^{z}\left(\overline{v}_{i}\prod\limits_{j=i+1}^{z}c_{j}\right)}{\sum\limits_{i=1}^{n}\left(\overline{v}_{i}\prod\limits_{j=1}^{i-1}c_{j}\right)} 	\\																
 &	= \dfrac{\sum\limits_{i=1}^{z}\left(\overline{v}_{i}\prod\limits_{j=i+1}^{z}c_{j}\right)}{\sum\limits_{i=1}^z\left(\overline{v}_{i}\prod\limits_{j=1}^{i-1}c_{j}\right)+ \prod\limits_{j=1}^{z}c_{j}\sum\limits_{i=z+1}^{n}\left(\overline{v}_{i}\prod\limits_{j=z+1}^{i-1}c_{j}\right)}, 
\end{align*}
\noindent by applying Lemma~\ref{lemma:sumoveri}, we have:
\begin{align*}
r&\geq \dfrac{\sum\limits_{i=1}^{z}\left(\overline{v}_{i}\prod\limits_{j=i+1}^{z}c_{j}\right)}{\sum\limits_{i=1}^{z}\left(\overline{v}_{i}\prod\limits_{j=1}^{i-1}c_{j}\right)+ \prod\limits_{j=1}^{z}c_{j}\sum\limits_{i=z+1}^{n}\left(\overline{v}_{i}\prod\limits_{j=z+1}^{i-1}c_{j}\right)} \\
& \geq \dfrac{\sum\limits_{i=1}^{z}\left(\overline{v}_{i}\prod\limits_{j=i+1}^{z}c_{j}\right)}{\sum\limits_{i=1}^{z}\left(\overline{v}_{i}\prod\limits_{j=1}^{i-1}c_{j}\right)+ \prod\limits_{j=1}^{z}c_{j}\dfrac{\overline{v}_{z+1}}{1-c_{{z+1}}}},							
\end{align*}
\noindent by using Equation~(\ref{eq:condizionikplus1}) we have:
\begin{align*}
r&\geq \dfrac{\sum\limits_{i=1}^{z}\left(\overline{v}_{i}\prod\limits_{j=i+1}^{z}c_{j}\right)}{\sum\limits_{i=1}^{z}\left(\overline{v}_{i}\prod\limits_{j=1}^{i-1}c_{j}\right)+ \prod\limits_{j=1}^{z}c_{j}\dfrac{\overline{v}_{{z+1}}}{1-c_{{z+1}}}}						\\		&\geq \dfrac{\sum\limits_{i=1}^{z}\left(\overline{v}_{i}\prod\limits_{j=i+1}^{z}c_{j}\right)}{\sum\limits_{i=1}^{z}\left(\overline{v}_{i}\prod\limits_{j=1}^{i-1}c_{j}\right)+ \prod\limits_{j=1}^{z}c_{j} \sum\limits_{i=1}^{z}\left(\overline{v}_{i}\prod\limits_{j=i+1}^{z}c_{j}\right)}.
\end{align*}
\noindent We derive $r$ w.r.t. $\overline{v}_{z}$ and we show that the derivative is always positive:
\begin{align*}
\dfrac{\partial r}{\partial \overline{v}_{z}}&= 
\dfrac{\sum\limits_{i=1}^{z}\left(\overline{v}_{i}\prod\limits_{j=1}^{i-1}c_{j}\right) - \prod\limits_{j=1}^{z-1}c_{j}\sum\limits_{i=1}^{z}\left(\overline{v}_{i}\prod\limits_{j=i+1}^{z}c_{j}\right) }{\left(\sum\limits_{i=1}^{z}\left(\overline{v}_{i}\prod\limits_{j=1}^{i-1}c_{j}\right)+ \prod\limits_{j=1}^{z}c_{j} \sum\limits_{i=1}^{z}\left(\overline{v}_{i}\prod\limits_{j=i+1}^{z}c_{j}\right)\right)^2}\\
&=\dfrac{\sum\limits_{i=1}^{z}\left(\overline{v}_{i} \prod\limits_{j=1}^{i-1}c_{j} \left( 1 - \prod\limits_{j=i}^{z-1}c_{j} \prod\limits_{j=i+1}^{z}c_{j}\right) \right)}{\left(\sum\limits_{i=1}^{z}\left(\overline{v}_{i}\prod\limits_{j=1}^{i-1}c_{j}\right)+ \prod\limits_{j=1}^{z}c_j \sum\limits_{i=1}^{z}\left(\overline{v}_{i}\prod\limits_{j=i+1}^{z}c_{j}\right)\right)^2}\\
&\geq 0,
\end{align*}
\noindent given that $\left( 1 - \prod\limits_{j=i}^{z-1}c_{j} \prod\limits_{j=i+1}^{z}c_{j}\right)\geq 0$. Thus, we minimize $r$ by minimizing $\overline{v}_{z}$. Due to Constraint~(\ref{eq:condizionikminus}) with $h=z-1$, the minimum feasible value of $\overline{v}_{z}$ is:
\[
\overline{v}_{z} = (1-c_{z})\left(\sum\limits_{i=1}^{z-1}\left(\overline{v}_{i}\prod\limits_{j=i+1}^{z-1}c_{j}\right)\right),
\]
\noindent substituting $\overline{v}_{z}$ in $r$ we have:
\begin{align*}
r&\geq 
\dfrac{\sum\limits_{i=1}^{z-1}\left(\overline{v}_{i}\prod\limits_{j=i+1}^{z-1}c_{j}\right)}{\sum\limits_{i=1}^{z-1}\left(\overline{v}_{i}\prod\limits_{j=1}^{i-1}c_{j}\right)+ \prod\limits_{j=1}^{z-1}c_{j} \sum\limits_{i=1}^{z-1}\left(\overline{v}_{i}\prod\limits_{j=i+1}^{z-1}c_{j}\right)}.
\end{align*}
\noindent By applying iteratively the above steps for $\overline{v}_{h}$ with $h$ from $h=z-2$ to $h=2$, we obtain the same bound of the case with $z=1$:
\[
r \geq \dfrac
{
\overline{v}_{1}
}
{
\overline{v}_{1} + c_{1} \overline{v}_{1}
}\geq \frac{1}{2}.
\]

\textbf{Bound with $z=k$}. Differently from the previous cases, here we do not use an upper bound over $\theta^*$, but the exact value, and for the lower bound over $\theta^*_{\prec_{\textrm{ads}}^R}$ we consider only the ads appearing in $\theta^*$. In this case, we have:
\[
r \geq \dfrac{\sum\limits_{i=1}^{k}\left(\overline{v}_{i}\prod\limits_{j=i+1}^{k}c_{j}\right)}{\sum\limits_{i=1}^{k}\left(\overline{v}_{i}\prod\limits_{j=1}^{i-1}c_{j}\right)},
\]
\noindent we derive $r$ w.r.t. $\overline{v}_{k}$, obtaining: 
\begin{align*}
\dfrac{\partial r}{\partial \overline{v}_{k}}			&		=			 
\dfrac{\sum\limits_{i=1}^{k}\left(\overline{v}_{i} \prod\limits_{j=1}^{i-1}c_{j} \left( 1 - \prod\limits_{j=i}^{k-1}c_{j} \prod\limits_{j=i+1}^{k}c_{j}\right) \right)}{\left(\sum\limits_{i=1}^{k}\left(\overline{v}_{i}\prod\limits_{j=1}^{i-1}c_{j}\right)\right)^2}\\
&\geq 0,
\end{align*}
\noindent we minimize $r$ by minimizing $\overline{v}_{k}$ as:
\[
\overline{v}_{k} = (1-c_{k})\left(\sum\limits_{i=1}^{k-1}\left(\overline{v}_{i}\prod\limits_{j=i+1}^{k-1}c_{j}\right)\right),
\]
\noindent obtaining:
\[
r \geq \dfrac{\sum\limits_{i=1}^{k-1}\left(\overline{v}_{i}\prod\limits_{j=i+1}^{k-1}c_{j}\right)}{\sum\limits_{i=1}^{k-1}\left(\overline{v}_{i}\prod\limits_{j=1}^{i-1}c_{j} \left( 1  + (1-c_{k})c_{i}\prod\limits_{j=i+1}^{k-1}c_{j}^2\right)\right)}.
\]
\noindent By applying iteratively the above steps for $\overline{v}_{h}$ with $h$ from $h=k-2$ to $h=2$---the derivative is always positive---, we obtain:
\[
r \geq \dfrac{\overline{v}_{1}}{\overline{v}_{1}+ \overline{v}_{1} c_{1} \left(1-\prod\limits_{i=2}^k c_{j}\right)} = \dfrac{1}{1+ c_{1} \left(1-\prod\limits_{i=2}^k c_{j}\right)},
\]
\noindent $r$ is minimized when $c_{1}=1$ and $\prod\limits_{i=2}^k c_{j} = 0$, obtaining \[r\geq \frac{1}{2}.\]
\end{proof}
		
The extension of the above result to the case in which $\lambda_k=\lambda$ for every $k$ is straightforward. Indeed, it is sufficient to observe that in this case we can redefine the continuation probabilities as $c_i'=c_i\lambda$ for every $i$ and subsequently redefine $\lambda_k=1$ for every $k$.
	\vfill
			
\newpage
	\subsection{Time--approximation tradeoffs}
		\begin{tabular}{c}
			\includegraphics[width=1\linewidth]{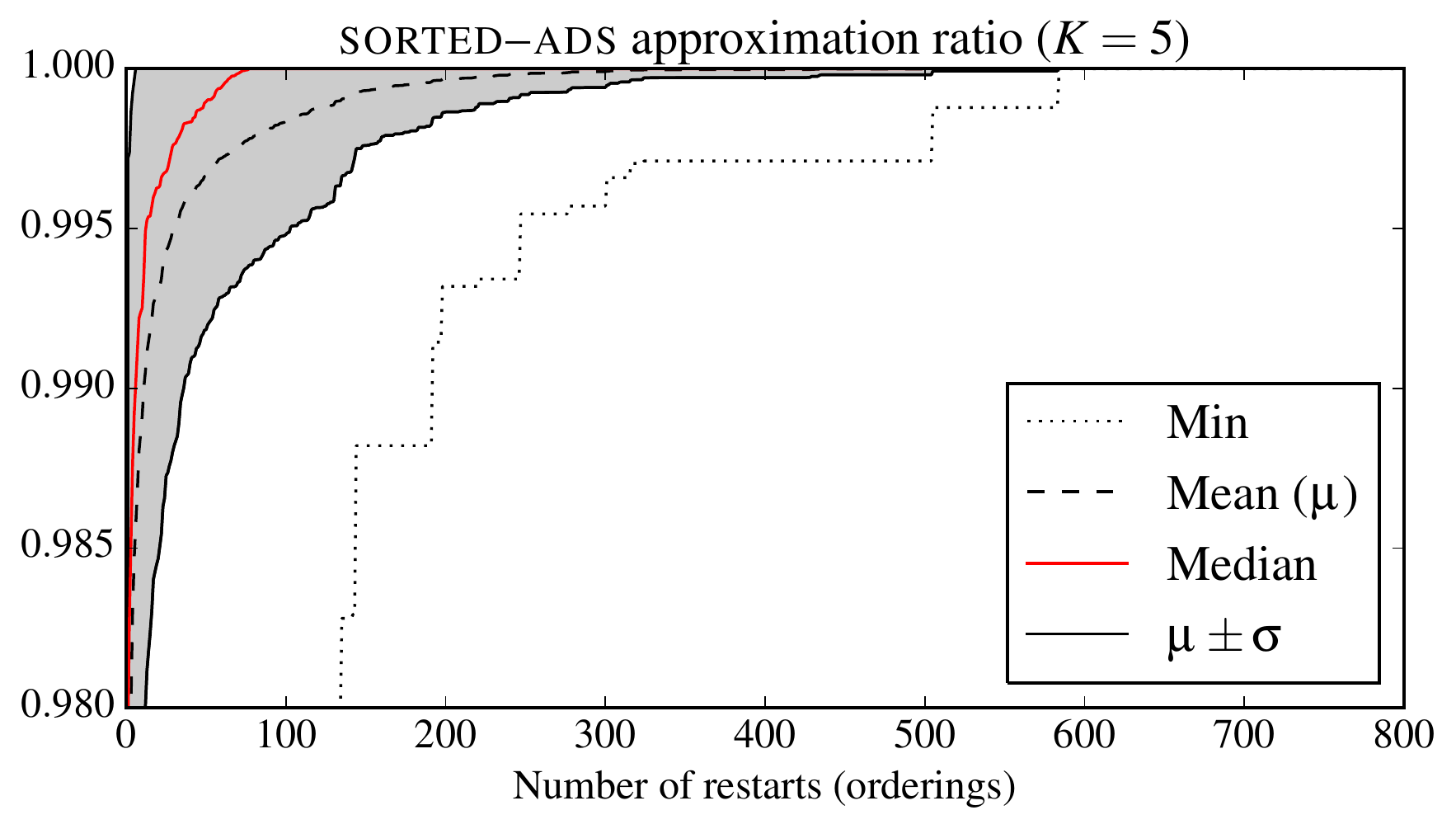}\\
			\includegraphics[width=1\linewidth]{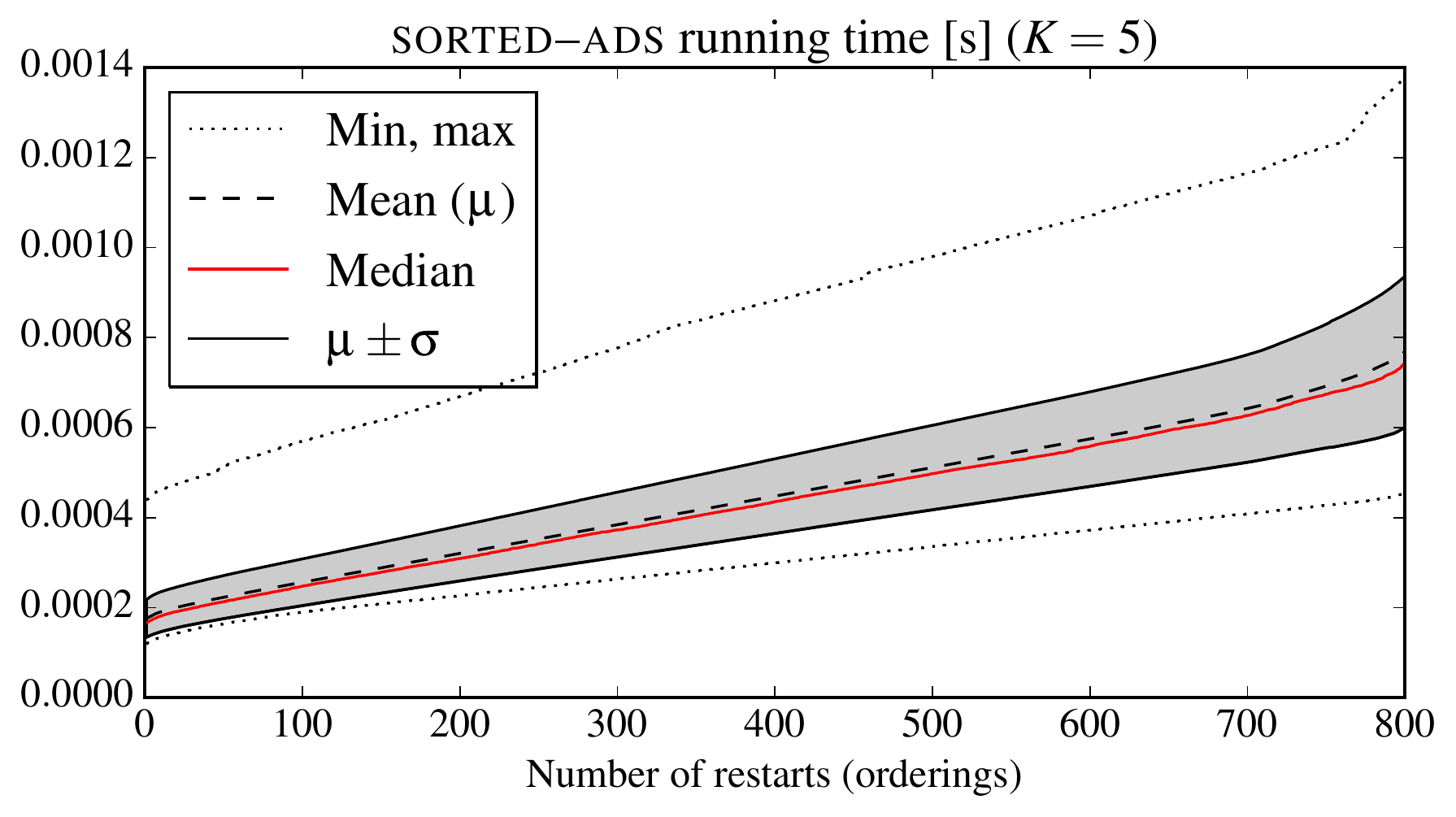}\\
			\includegraphics[width=1\linewidth]{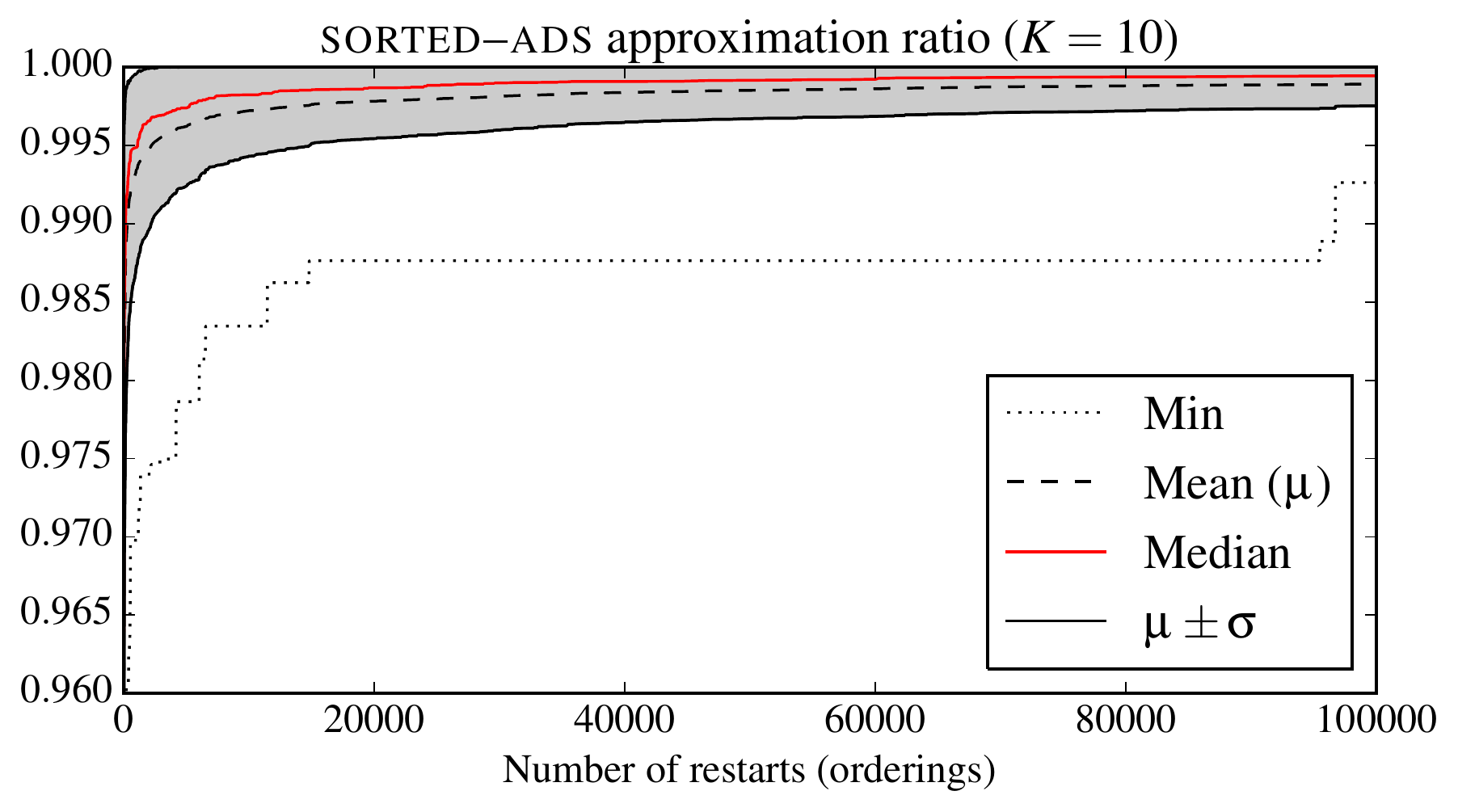}\\
			\includegraphics[width=1\linewidth]{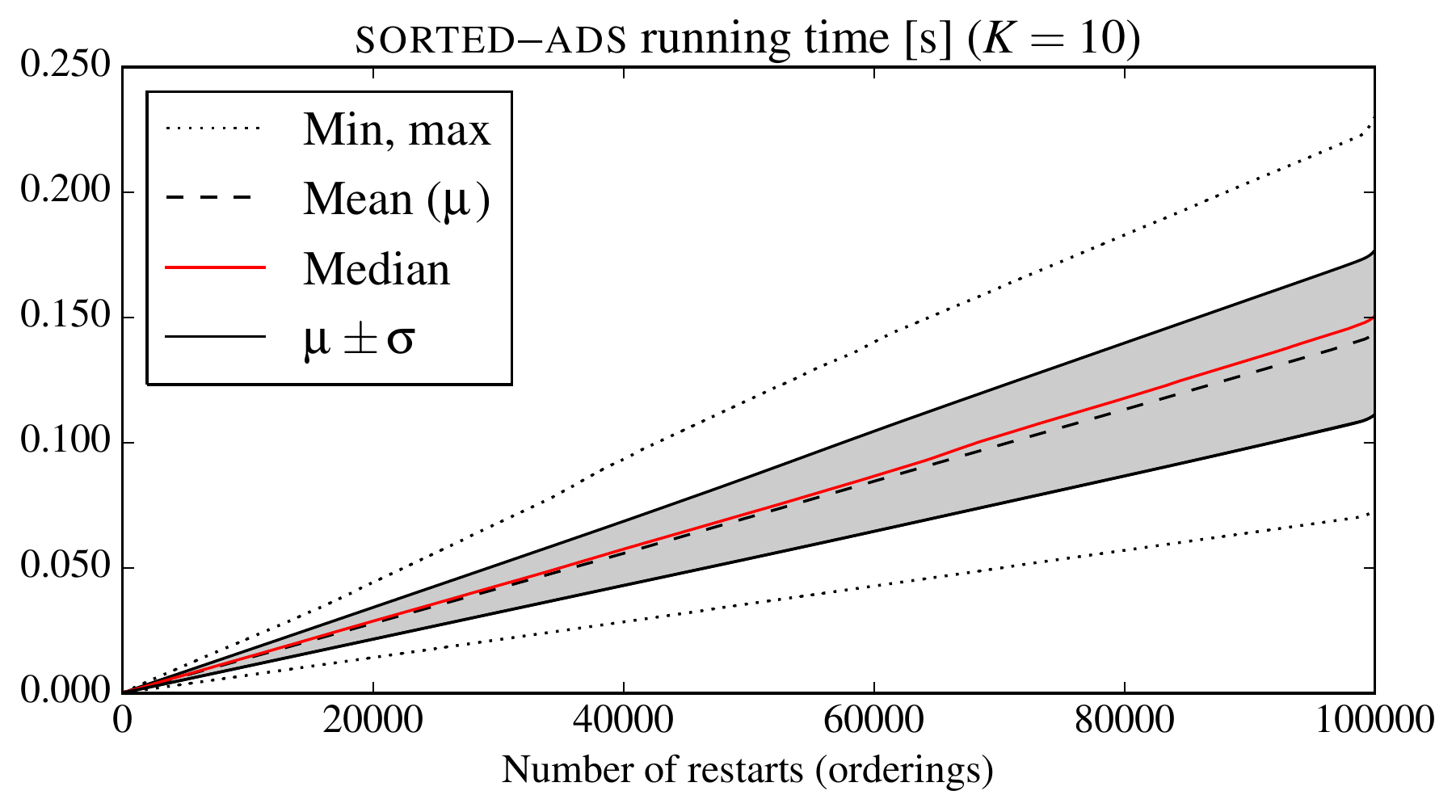}
		\end{tabular}
	
\newpage
\section{Experimental evaluation}
	We prove the statistical significance of the experimental results in the following boxplots. We follow the structure of the paper.
	\subsection{\textsc{dominated--ads} algorithm}
		\begin{tabular}{c}
			\noindent\includegraphics[width=1\linewidth]{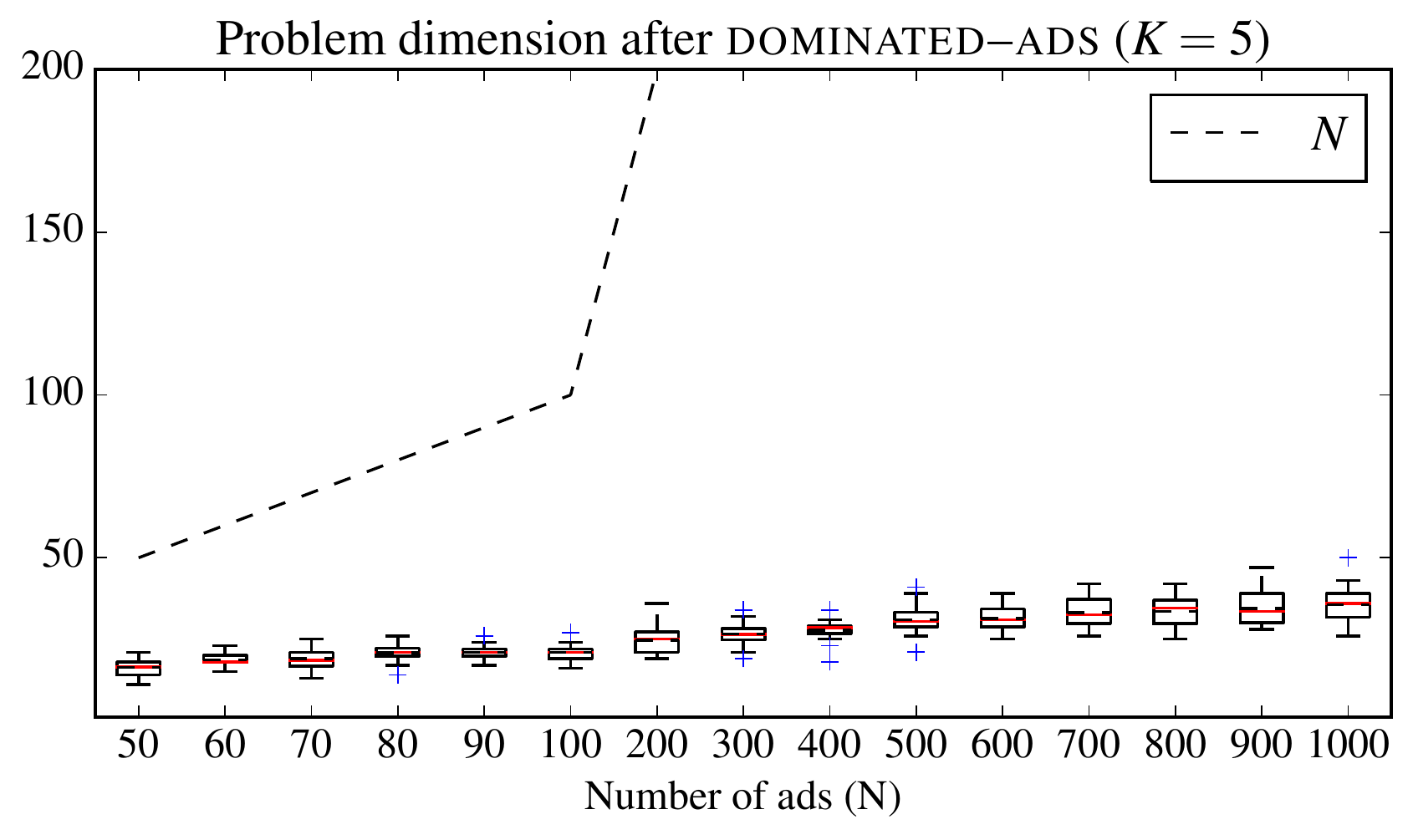}\\
			\noindent\includegraphics[width=1\linewidth]{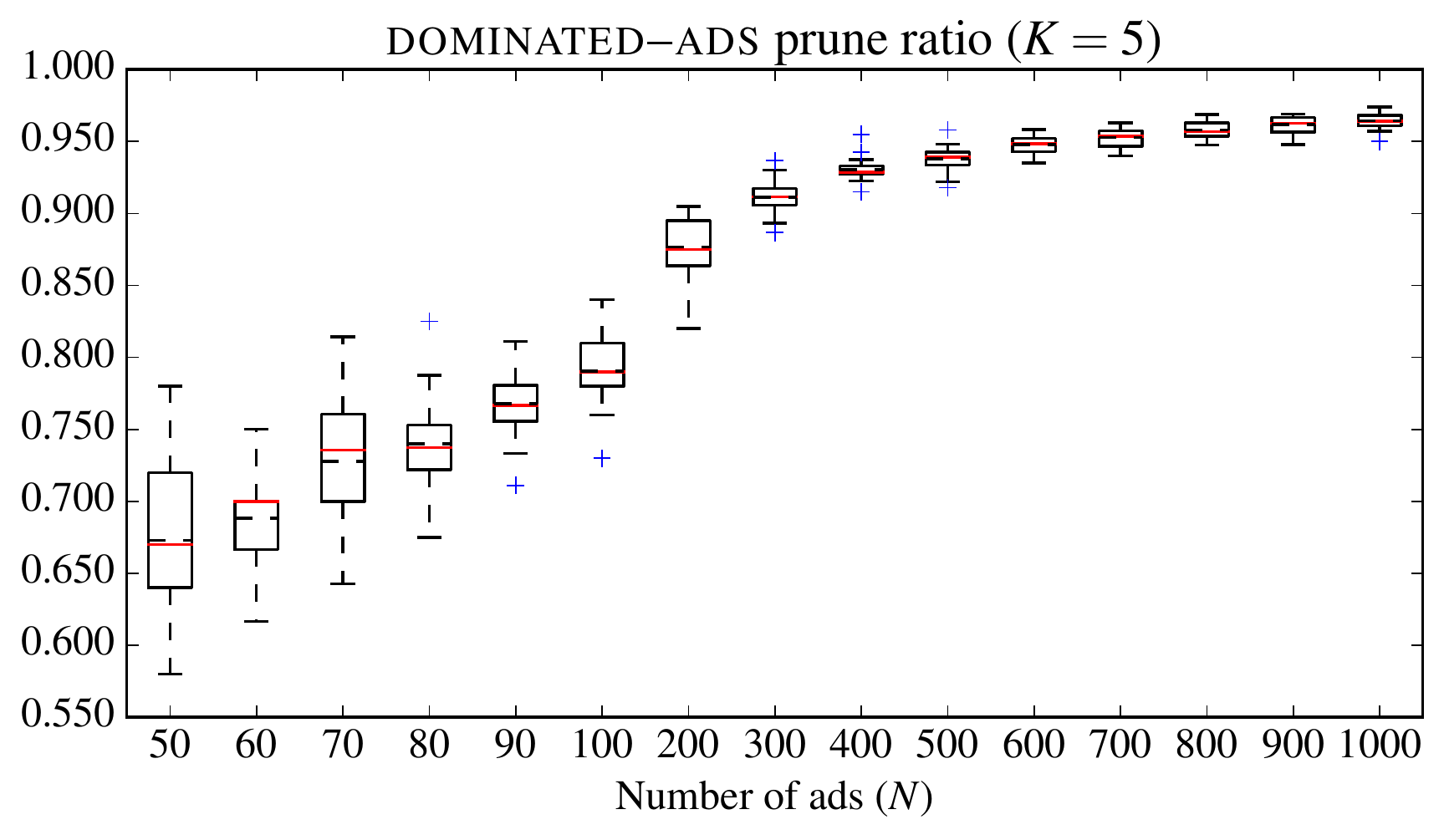}\\
			\noindent\includegraphics[width=1\linewidth]{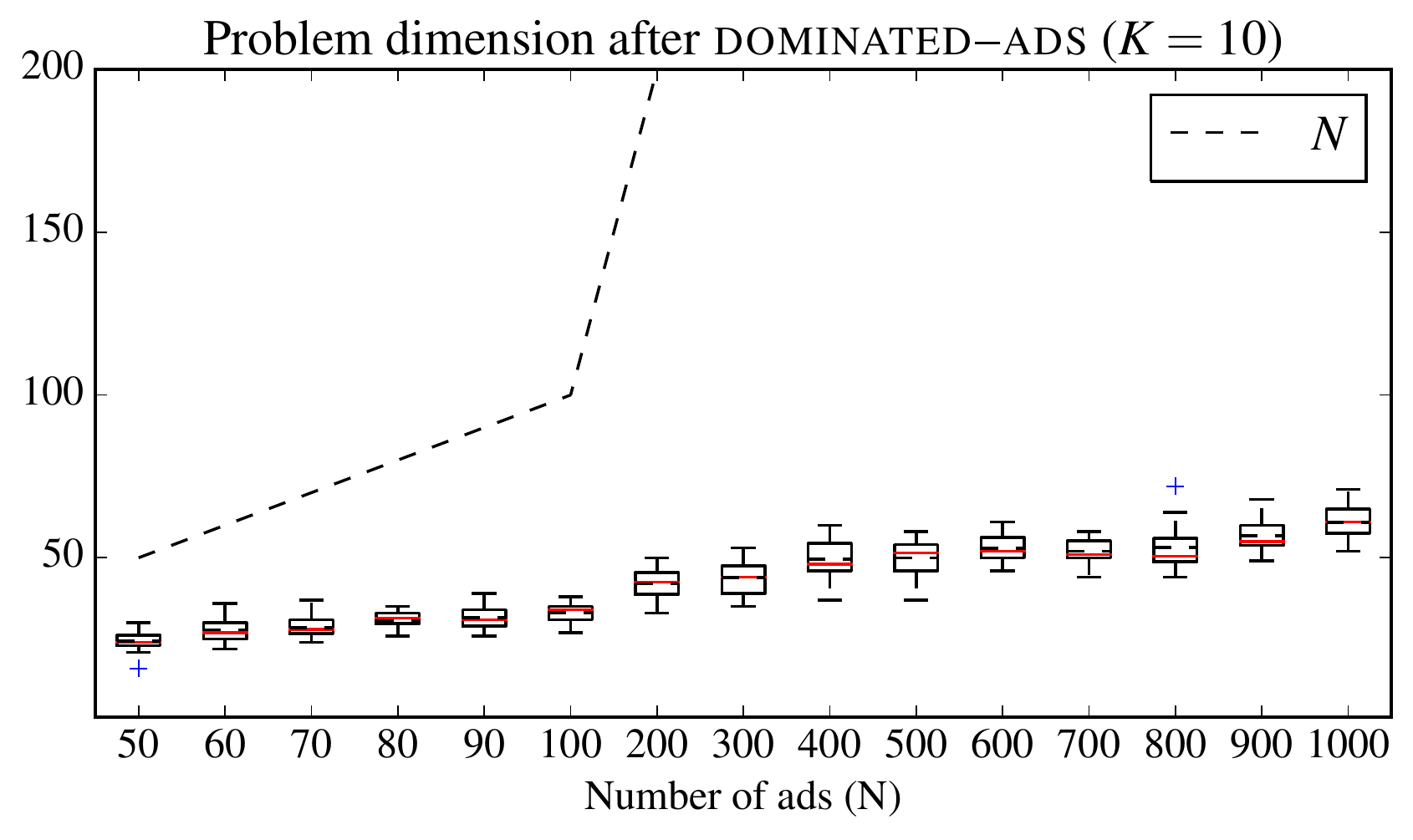}\\
			\noindent\includegraphics[width=1\linewidth]{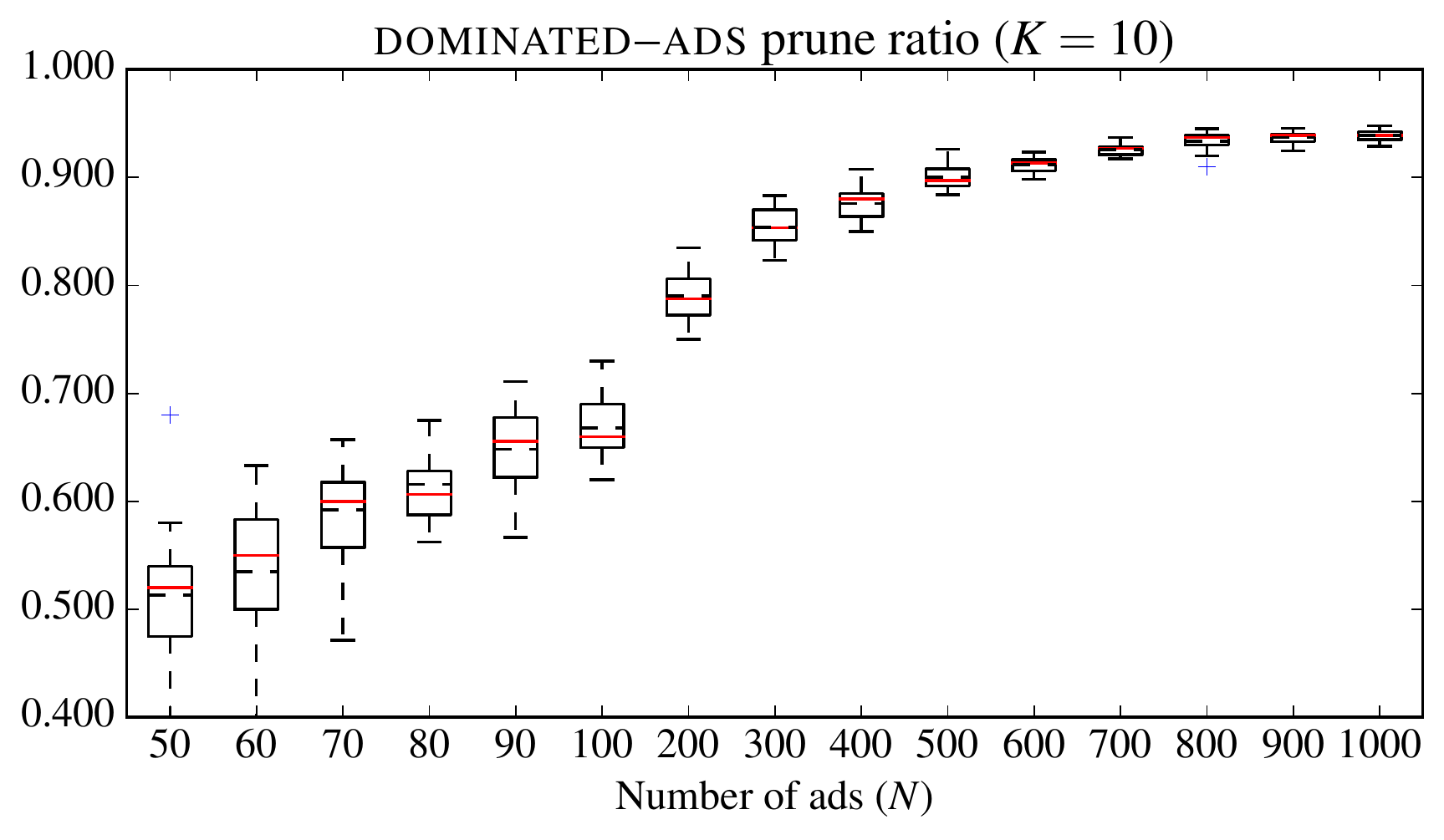}
		\end{tabular}

		\begin{tabular}{c}
			\noindent\includegraphics[width=.98\linewidth]{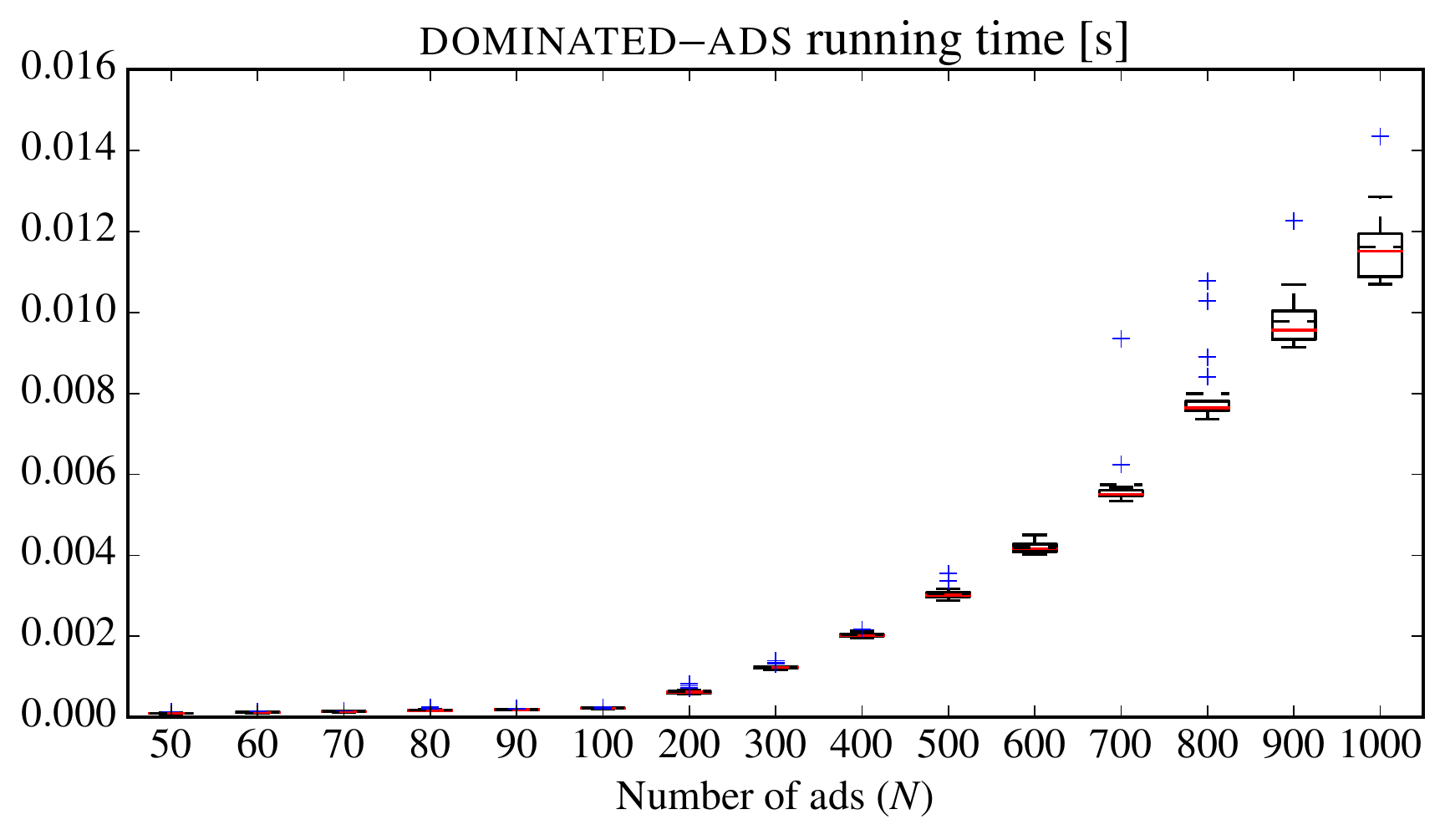}
		\end{tabular}
	\subsection{\textsc{colored--ads} algorithm}
		\begin{tabular}{c}
			\includegraphics[width=1\linewidth]{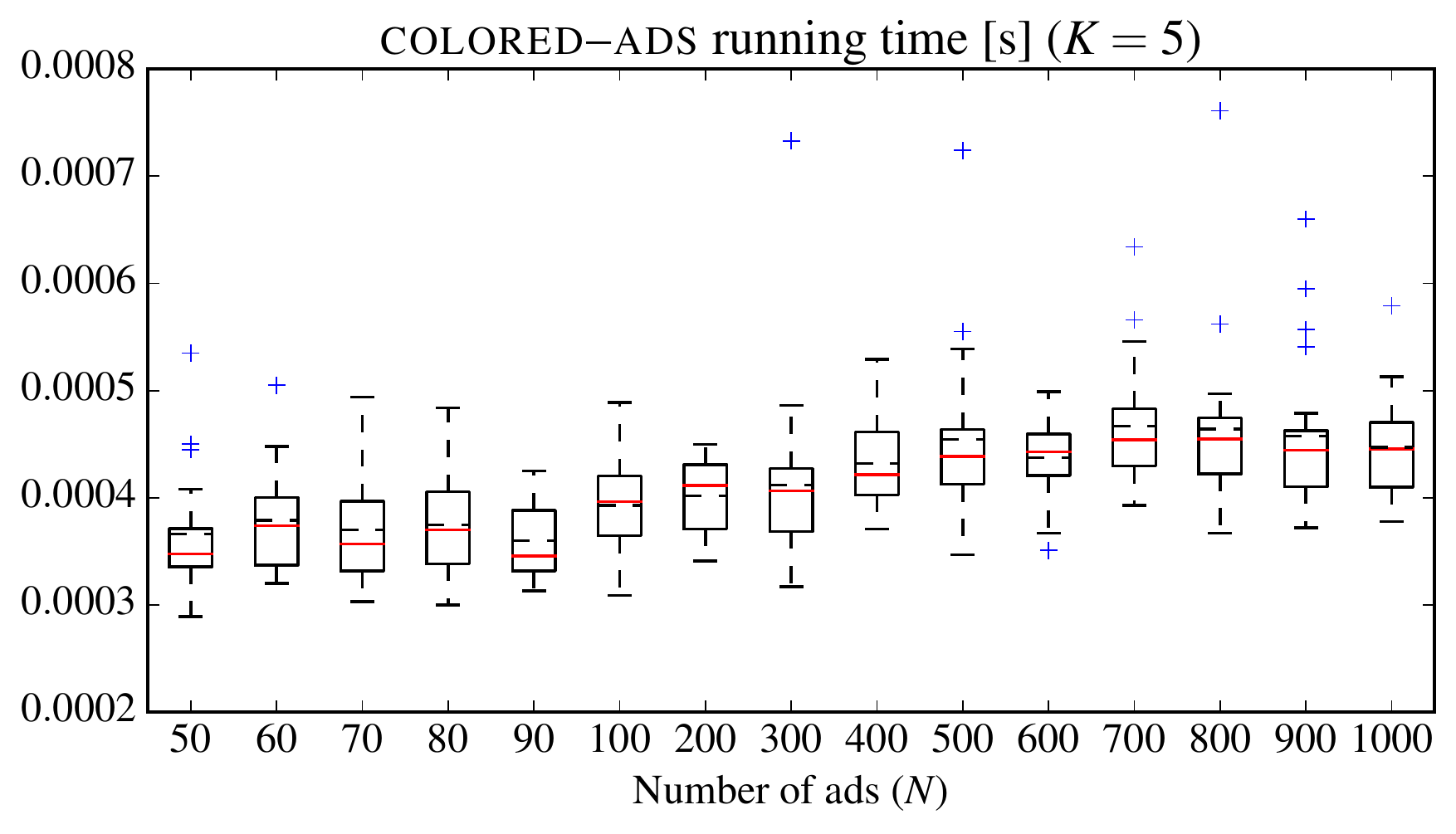}\\
			\includegraphics[width=1\linewidth]{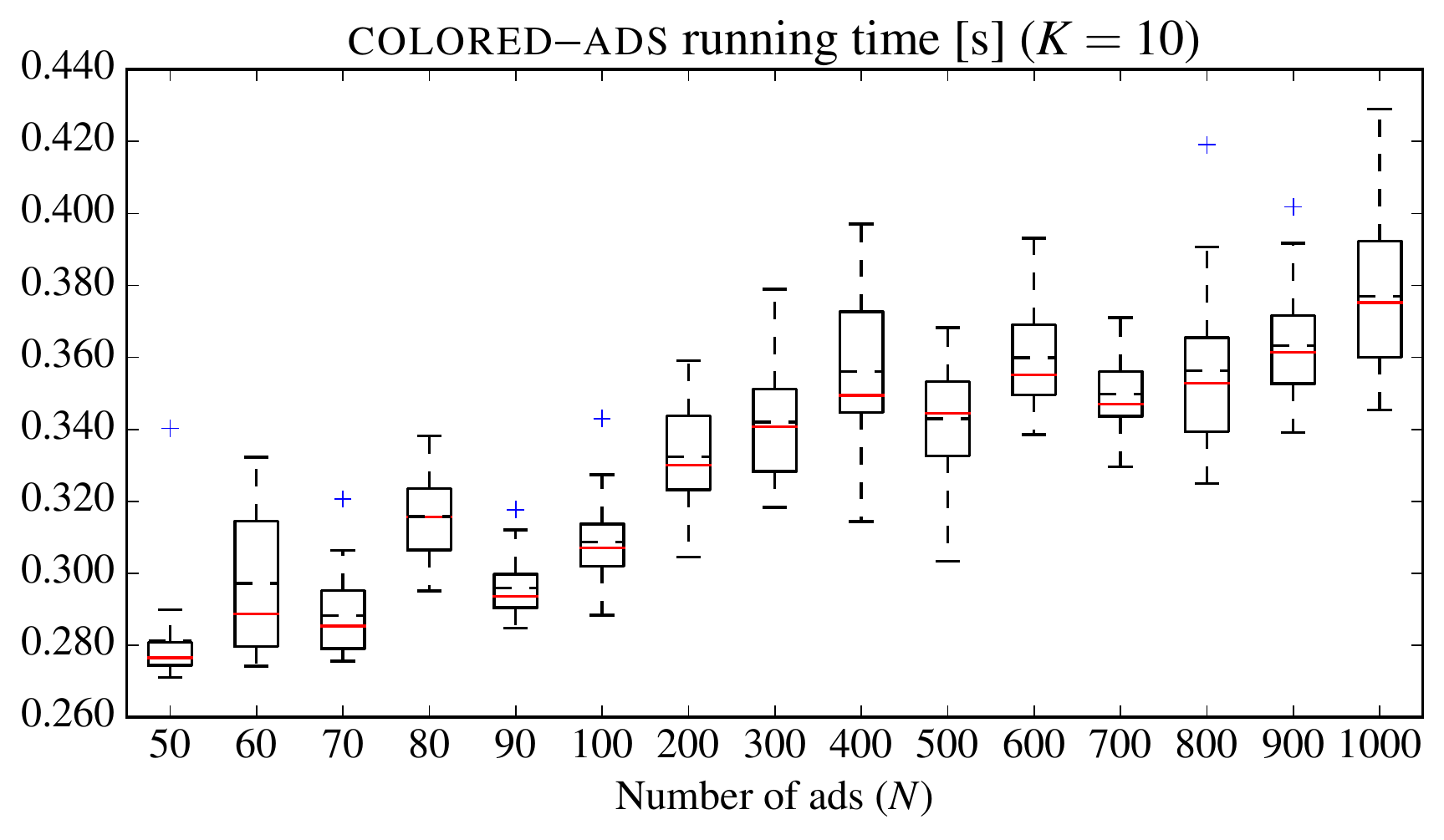}
		\end{tabular}
	
	\newpage
	\subsection{\textsc{sorted--ads} algorithm}
		\begin{tabular}{c}
			\includegraphics[width=1\linewidth]{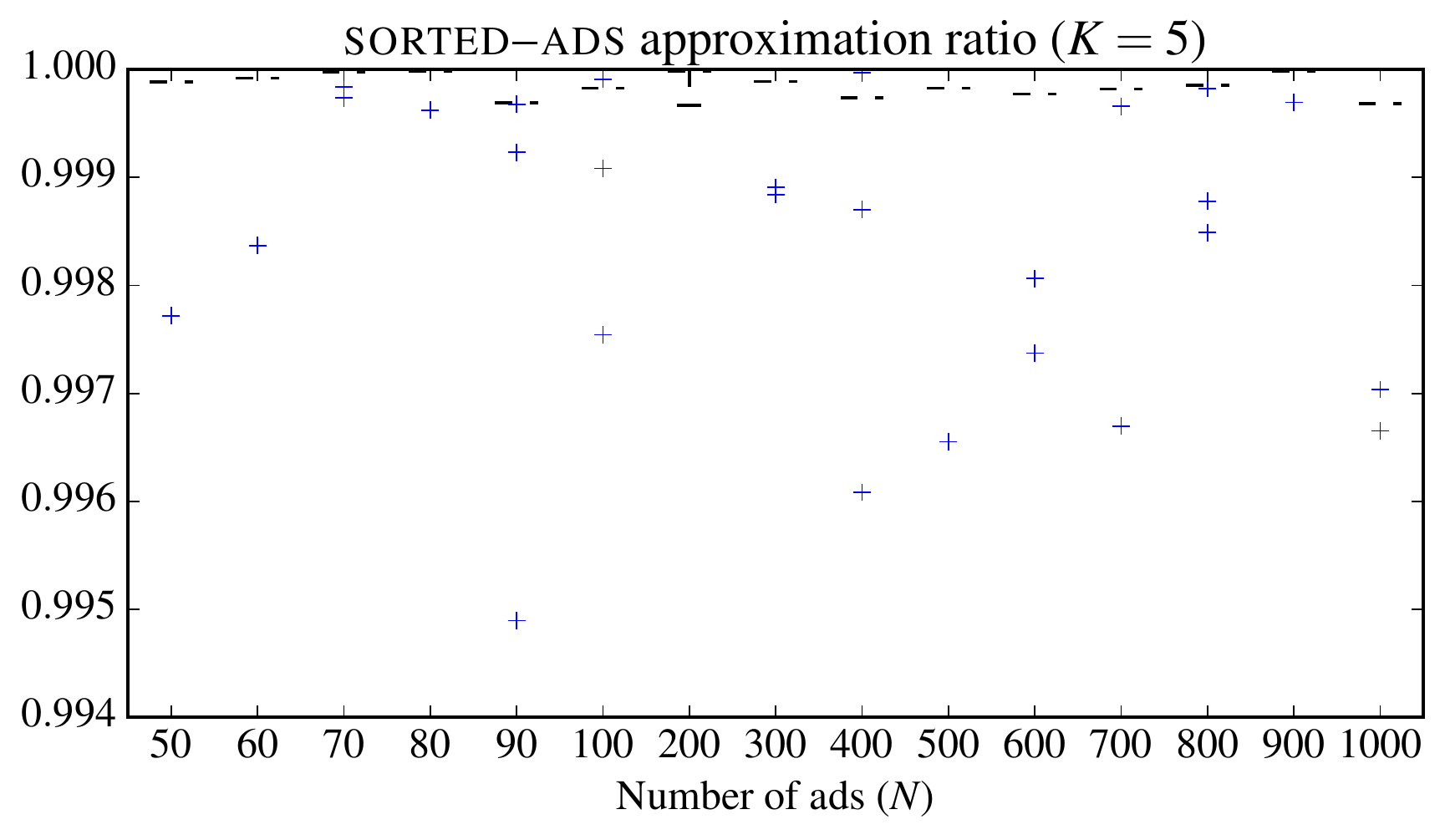}\\
			\includegraphics[width=1\linewidth]{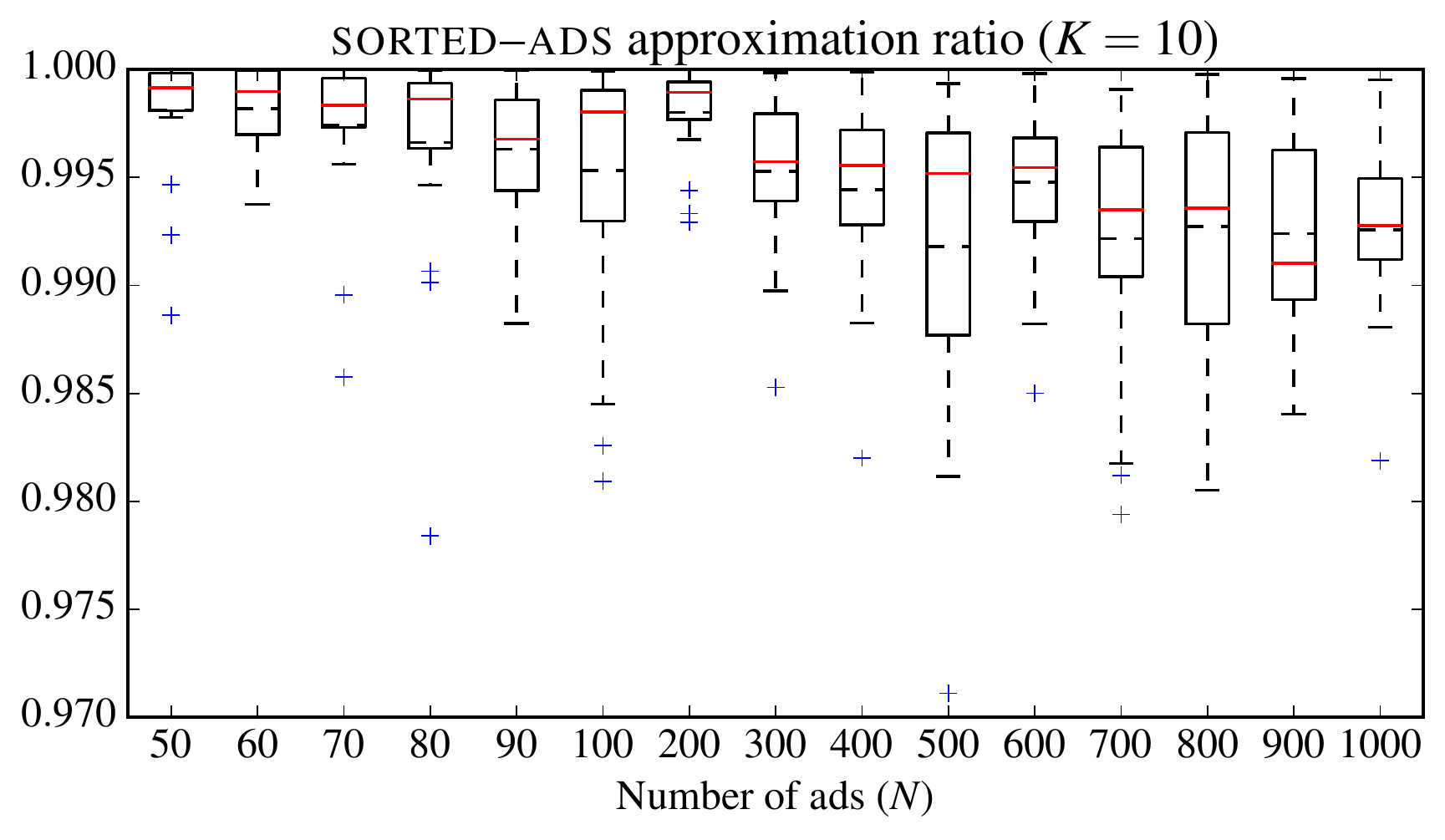}\\
			\includegraphics[width=1\linewidth]{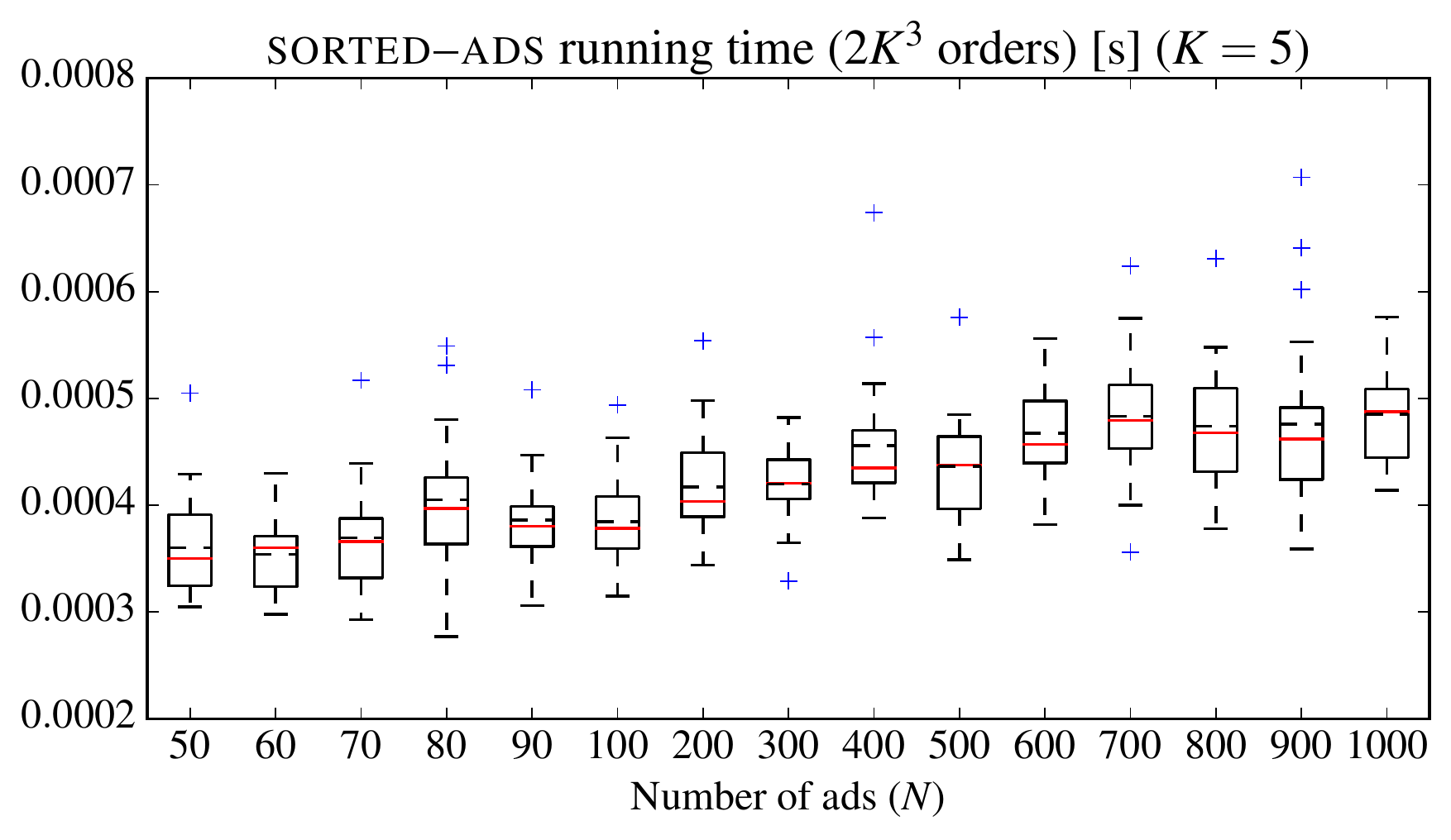}\\
			\includegraphics[width=1\linewidth]{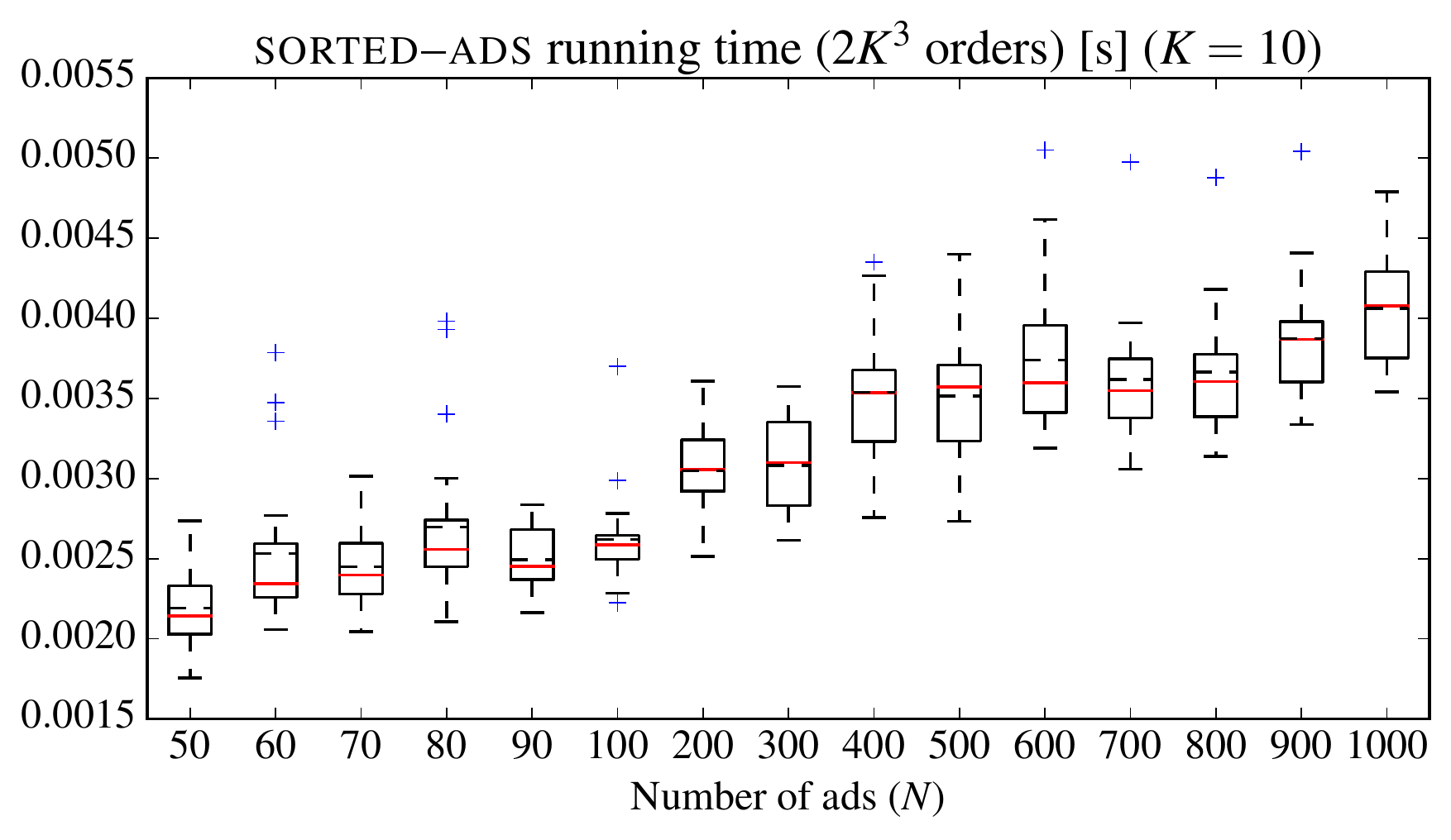}
		\end{tabular}
		
\end{appendices}
\fi
\end{document}